\documentclass[a4paper, 11pt]{amsart}

\usepackage{caption}
\usepackage[english]{babel} 
\usepackage[utf8]{inputenc}
\usepackage[T1]{fontenc}
\usepackage{lmodern}
\usepackage{setspace}

\usepackage{enumerate}
\usepackage{mathtools}
\usepackage{mathrsfs}
\usepackage{amsmath}
\usepackage{amsthm}   
\usepackage{amsfonts}
\usepackage{amssymb}
\usepackage{textcomp} 
\usepackage{wasysym}
\usepackage{color}

\usepackage{graphicx}  

\usepackage[textheight=23 cm, textwidth=15.5cm, headheight=50pt, headsep=10pt, footskip=32pt, centering]{geometry}
\usepackage[english]{babel}
\usepackage{dsfont}
\usepackage{cases}
\usepackage{comment}
\usepackage{float}
\usepackage{nicefrac}
\usepackage{algorithm,algorithmic}
\usepackage[colorlinks=true,linkcolor=black,citecolor=black]{hyperref}
\usepackage[absolute]{textpos}
\usepackage{placeins}
\usepackage{soul}
\usepackage{yhmath}

\theoremstyle{plain}
\newtheorem{theorem}[subsubsection]{Theorem}
\newtheorem{lemma}[subsubsection]{Lemma}

\newtheorem{assumption}[subsubsection]{Assumption}
\newtheorem{proposition}[subsubsection]{Proposition}
\newtheorem{remark}[subsubsection]{Remark}
\theoremstyle{definition}
\newtheorem{definition}[subsubsection]{Definition}
\newtheorem{example}[subsubsection]{Example}

\DeclareMathOperator{\interior}{int}
\DeclareMathOperator{\dom}{dom}

\DeclareMathOperator{\VaR}{VaR}
\DeclareMathOperator{\LVaR}{LVaR}
\DeclareMathOperator{\RVaR}{RVaR}
\DeclareMathOperator{\ES}{ES}
\DeclareMathOperator{\symbolSCRM}{SC}
\DeclareMathOperator{\symbolCRM}{C}
\DeclareMathOperator{\symbolFCRM}{FC}
\DeclareMathOperator{\symbolAERM}{AE}

\newcommand{\Rho}{\mathcal{P}}
\newcommand{\adjustedRiskMeasure}{\rho_{\Rho,g}}
\newcommand{\adjustedRiskMeasureInfimal}{\rho_{\Rho,ng}}
\newcommand{\scrm}{\rho^{\symbolSCRM}_{r,g}}
\newcommand{\crm}{\rho^{\symbolCRM}_{\mathcal{L},g}}
\newcommand{\fcrm}{\rho^{\symbolFCRM}_{\mathcal{L},g}}
\newcommand{\aerm}{\rho_{g}^{\symbolAERM}}

\DeclareMathOperator*{\esssup}{ess\,sup}
\DeclareMathOperator*{\essinf}{ess\,inf}

\title{Risk measures based on target risk profiles}

\usepackage[foot]{amsaddr}

\author{Jascha Alexander\,$^{1,\dagger}$}
\author{Christian Laudagé\,$^{1,\ast}$}
\author{Jörn Sass\,$^{1,\ddagger}$}
\thanks{$^1$ Department of Mathematics, RPTU Kaiserslautern-Landau, Gottlieb-Daimler-Straße 47, 67663 Kaiserslautern, Germany}

\thanks{$^\dagger$ \textit{E-mail address:} \texttt{jalexand@rptu.de}}
\thanks{$^{\ast}$  \textit{E-mail address:} \texttt{christian.laudage@rptu.de} (corresponding author)}
\thanks{$^\ddagger$ \textit{E-mail address:} \texttt{joern.sass@rptu.de}}

\date{\today}

\usepackage{fancyhdr}
\begin{document}
    \pagenumbering{arabic}
	\numberwithin{equation}{section}\numberwithin{equation}{section}
	
	\begin{abstract}
    We address the problem that classical risk measures may not detect the tail risk adequately. This can occur for instance due to averaging when calculating the Expected Shortfall. The current literature proposes the so-called adjusted Expected Shortfall as a solution. This risk measure is the supremum of Expected Shortfalls for all possible levels, adjusted with a function $g$, the so-called target risk profile. We generalize this idea by using a family of risk measures which allows for more choices than Expected Shortfalls, leading to the concept of adjusted risk measures. An adjusted risk measure quantifies the minimal amount of capital that has to added to a financial position $X$ to ensure that each risk measure out of the chosen family is smaller or equal to the target risk profile $g(p)$ for the corresponding level $p\in[0,1]$. We discuss a variety of assumptions such that desirable properties for risk measures are satisfied in this setup. From a theoretical point of view, our main contribution is the analysis of equivalent assumptions such that an adjusted risk measure is positive homogeneous and subadditive. Furthermore, we show that these conditions hold for several adjusted risk measures beyond the adjusted Expected Shortfall. In addition, we derive their dual representations. Finally, we test the performance in a case study based on the S$\&$P~$500$. 

    \medskip
	
    \item[\hskip\labelsep\scshape Keywords:] Adjusted Expected Shortfall, expectiles, Loss Value-at-Risk, Range Value-at-Risk, target risk profile  
    	
    \medskip
	
    \item[\hskip\labelsep\scshape JEL Classification:] G11, G32

    \end{abstract}
    \maketitle
    
    \section{Introduction}
	Classical risk measures, like Value-at-Risk (VaR) and Expected Shortfall (ES), are commonly used in finance and insurance. The VaR is a quantile and it is interpreted as the minimal amount of cash needed such that bankruptcy only occurs with a pre-specified probability level. Unlike VaR, ES captures the tail behavior by averaging over it. However, due to the simplicity of the averaging process, ES is not able to capture every tail behavior appropriately. In fact, we could have two different distributions with an equal risk value by ES, but they differ significantly by the probability mass lying in the tail of the distribution. This then leads to an underestimation of risk by the ES. To fix this problem, the literature recently introduced the Loss VaR and the adjusted ES. Both are based on the idea to calculate VaR, respectively ES, at every probability level and adjust those values by a given function, the so-called target risk profile. The supremum of these values yields the Loss VaR and the adjusted ES, respectively. In this work, we generalize this idea.
	
	\textit{Literature review:} VaR was introduced under the name RiskMetrics in $1994$ by JP Morgen. For the history of the development of the VaR we refer to \cite{guldimann2000story}. Thereafter, the focus shifted to ES (also called Conditional VaR or Average VaR), for early conclusions we refer to \cite{ACERBI20021487}, \cite{FREY20021317} and \cite{ROCKAFELLAR20021443}. For implications and deficits of VaR and ES, concerning the tail behavior, we refer to \cite{https://doi.org/10.1111/1467-9965.00068,Wang,RVaR2,WEBER2018191}.
 
    Driven by the Basel accords and Solvency II, a debate about the benefits and the drawbacks of VaR and ES erupted, see e.g.,~\cite{VaRgegenES1} and \cite{VaRgegenES2}. As a consequence, the ES is often considered as more appropriate than the VaR and the Basel Committee agreed on using the ES, see \cite{bcbs2012consultative}. But, there are also disadvantages of the ES, see e.g.,~\cite{KMM}. \cite{Robust} point out that the VaR (in contrast to the ES) satisfies the classical Hampel robustness, see \cite{hampel2011robust}. Contributing to this debate, \cite{Kraetschmer2014} argue that Hampel robustness is not an adequate choice to define the robustness of a risk measure. Hence, they use a different metric for which the ES turned out to be robust. This hypothesis is also confirmed by~\cite{doi:10.1287/opre.2021.2147}, who investigate robustness after a preceding portfolio optimization. 
    
    Our starting point is the idea suggested in~\cite{LVaR}, which aims to capture the tail behavior of a distribution more adequately. To do so, for every probability level, the corresponding VaR value is reduced by the value of the target risk profile. The supremum with respect to these reduced VaR values is the so-called Loss VaR (LVaR).
    
    The methodology in \cite{AES} follows this idea, but instead of using a family based on VaR they use a family based on ES. This results in the so-called adjusted ES. \cite[Theorem 3.1]{Mao} states that any risk measure which is consistent with second-order stochastic dominance can be represented as an infimum of a family of adjusted ESs.

    Recently, following the idea of an adjusted ES, also other examples are proposed in which the family of ES is substituted by another concrete family of risk measures. For instance, \cite{ZOU2023255} make use of a family of entropic Value-at-Risks of order $p\in(-\infty,0)\cup(0,\infty]$. They call it an adjusted Rényi entropic VaR. Further, \cite{ZOU20241} generalize an adjusted ES by using a family of higher-order Expected Shortfalls. The resulting risk measure is called an adjusted higher-order ES.
	
	\textit{Methodology:} We generalize the aforementioned concepts by allowing in their definitions for an arbitrary family of risk measures, instead of restricting attention to a family of VaRs (\cite{LVaR}) or a family of ESs (\cite{AES}). To be precise, the adjusted ES relies on the ES values for all levels between zero and one. Instead, we allow to substitute the ES by other risk measures. The new risk measure is called an adjusted risk measure. It is the supremum of risk measures minus the value of a function $g$, the so-called target risk profile. To the best of our knowledge, the only reference that mentioned such a generalization of the adjusted ES is~\cite{herdegen2024rhoarbitragerhoconsistentpricingstarshaped}. However, the authors do not elaborate on properties of adjusted risk measures. We close this gap in the literature by providing a comprehensive analysis of adjusted risk measures.
 
    \textit{Main contributions:} Under a weak assumption on $g$, we state an equivalent condition for positive homogeneity and subadditivity of the adjusted risk measure. The assumption on $g$ states that the adjusted risk measure is finite and the supremum in the construction of the adjusted risk measure can be attained for a value for which the target risk profile $g$ is unequal zero. Afterwards, we analyze explicit constructions of adjusted risk measures. To do so, in contrast to the adjusted ES, which solely relies on Expected Shortfalls for levels between zero and one, we use different kinds of risk measures for different levels, like VaR, Range-VaR, ES or expectiles. We show that all of the resulting adjusted risk measures are neither positive homogeneous nor subbadditive in general. Furthermore, we develop dual representations of the new constructions.
    
    Finally, we evaluate the new risk measures in a case study based on the S$\&$P 500 and single stocks from this index. First, we test a step function $g$ as target risk profile, compare with~\cite{AES}. We find that in this case, the new risk measures react sensitively to even small changes in the underlying time series. Second, instead of step functions, we use target risk profiles that we obtain as a family of risk measures evaluated for the S$\&$P 500. Here, we test different underlying time periods to calibrate the target risk profile. For a time period in which the S$\&$P 500 is less volatile, we observe similar effects as in the case of a step function. In contrast, using a time period in which the S$\&$P 500 is more volatile, then the new risk measures are useful to visualize times of crisis in the underlying data. Lastly, we analyze the new risk measures by reevaluating the target risk profiles over time. In doing so, the new risk measures can underestimate times of crisis. From this we conclude that such a procedure should only be used by less risk averse investors. Additionally, we use the new risk measures to compare different financial markets. It turns out that such a procedure is useful to detect differences between these markets.
	
	\textit{Structure of the manuscript:} In Section $2$, the adjusted risk measure is defined and main properties are discussed. In particular, we introduce our main theorem regarding the positive homogeneity and subadditivity of the adjusted risk measure. Then, in Section $3$, we analyze explicit examples. In Section $4$ we test the behavior of several new risk measures in a case study based on the S$\&$P $500$ and individual stock shares. 

    \textit{Basic notations and definitions:} Let $(\Omega,\mathcal{F},\mathbb{P})$ be an atomless\footnote{This property is used to guarantee the statements in Table~\ref{table1} and is needed to construct counterexamples in  Appendix~\ref{sec:explanationsTable}.} probability space. In the following, we denote by $L^0(\Omega,\mathcal{F},\mathbb{P})$, or $L^0$ for short, the linear space of all equivalence classes with respect to $\mathbb{P}$-almost sure equivalence of real-valued random variables over $(\Omega,\mathcal{F},\mathbb{P})$. For $p\in[1,\infty)$, we denote by  $L^p(\Omega,\mathcal{F},\mathbb{P})$, or $L^p$ for short, the linear space of all equivalence classes of real-valued random variables with finite $p$-th moment. Further, we denote by $L^\infty(\Omega,\mathcal{F},\mathbb{P})$, or $L^\infty$ for short, the linear space of essentially bounded random variables. For $X,Y\in L^0$, we write $X\sim Y$, if $X$ and $Y$ are identically distributed. For $X,Y\in L^1$, we write $X\geq_{\text{SSD}}Y$, if $E[u(-X)]\geq E[u(-Y)]$ for every increasing and concave function $u:\mathbb{R}\rightarrow\mathbb{R}$. Let $\mathcal{Y}\subseteq L^0$. The epigraph of a map $f:\mathcal{Y}\rightarrow[-\infty,\infty]$ is $\text{epi}(f) \coloneqq \{(X,\alpha)\in \mathcal{Y}\times \mathbb{R}\,|\, f(X)\leq \alpha\}$.
    
	\section{Adjusted risk measure}
 
    In this manuscript, we work with the loss of financial positions at the end of a fixed time period. The linear space of all possible losses is denoted by $\mathcal{X} = L^{q}$ for $q\in\{0,1\}$. In particular, we use $\mathcal{X}$ as the domain for the risk measures. It should be clear that positive values of $X\in \mathcal{X}$ refer to losses. 
 
	As a last prerequisite, we introduce a set of functions, which we use as possible target risk profiles in the definition of the adjusted risk measure later on. For this, let $\mathcal{G}$ be the set of all increasing (in the non-strict sense) functions $g:[0,1] \rightarrow (-\infty,\infty]$ with $g(p)<\infty$ for at least one $p\in(0,1]$.

 \subsection{Definition of adjusted risk measures}\label{sec:def_adjustedRiskMeasure}
	Before we introduce our new concept, we recall the definition of the adjusted ES, as introduced in~\cite{AES}. For this, we work throughout the whole manuscript under the convention that $\infty-\infty\,\,\text{is set to}\,-\infty$.
	
	\begin{definition}[VaR, ES and adjusted ES]
		Let $Y \in L^0$ and $X\in L^1$, then we define
		\begin{equation*} 
			\text{VaR}_p(Y)\coloneqq 
			\begin{cases}
				\text{inf}\{y\in \mathbb{R} \, |\, \mathbb{P}(Y\leq y) \geq p\} & \text{if}\,\, p \in (0,1], \\ 
				\essinf Y & \text{if}\,\, p=0,
			\end{cases} 
		\end{equation*}
		
		\begin{equation*} 
			\text{ES}_p(X)\coloneqq 
			\begin{cases} 
				\frac{1}{1-p} \,$$\int_{p}^{1}\text{VaR}_q(X)\,\mathrm{d}q$$ & \text{if}\,\, p \in [0,1), \\ 
				\esssup X & \text{if}\,\, p=1. 
			\end{cases} 
			\\
		\end{equation*}
        
		Given a function $g\in\mathcal{G}$, the adjusted ES is defined as
		\begin{equation*}
			\text{ES}^g(X) \coloneqq \sup_{p\in [0,1]}\{\text{ES}_p(X)-g(p)\}.
		\end{equation*}
	\end{definition}

	In the following, we introduce a monetary risk measure and standard properties of it. To do so, note that a map $\rho:\mathcal{X}\rightarrow [-\infty,\infty]$ is called to be convex, subadditive, positively homogeneous, or star-shaped, whenever its epigraph $\text{epi}(\rho)$ is convex, closed under addition, a cone, or star-shaped, respectively. If the range of the map $\rho$ is a subset of $(-\infty,+\infty]$, then we get the following four statements: (1) $\rho$ is convex iff $\rho(\lambda X+(1-\lambda)Y)\leq \lambda\rho(X)+(1-\lambda)\rho(Y)\text{ for all }X,Y\in \mathcal{X}\text{ and }\lambda\in[0,1]$. (2) $\rho$ is  subadditive iff $\rho(X+Y)\leq \rho(X)+\rho(Y) \text{ for all } X,Y\in\mathcal{X}$. (3) $\rho$ is positive homogeneous iff $\rho(\lambda X) = \lambda \rho(X) \text{ for all } X\in \mathcal{X}\text{ and }  \lambda\geq 0.$ (4) $\rho$ is star-shaped iff $\rho(\lambda X)\geq \lambda \rho(X)\text{ for all }X\in\mathcal{X}\text{ and all } \lambda>1.$ The latter is equivalent to $\rho(\lambda X)\leq \lambda \rho(X)\text{ for all }X\in\mathcal{X}\text{ and all } \lambda\in[0,1].$ 
	\begin{definition}(Monetary risk measure)\label{def}
		A map $\rho:\mathcal{X}\rightarrow [-\infty,\infty]$ is called a \textbf{\textit{risk functional}}, if it satisfies the following three properties:		
		\begin{enumerate}
			\item[(i)] Monotonicity: $\rho(X) \leq \rho(Y) \text{ for all } X,Y\in \mathcal{X} \text{ with } X\leq Y,$
			\item[(ii)] Cash additivity: $\rho(X+m)=\rho(X)+m \text{ for all } X\in \mathcal{X}\text{ and }m\in \mathbb{R},$
		\end{enumerate}
        If $\rho>-\infty$, then we call $\rho$ a \textbf{\textit{(monetary) risk measure}}. A risk measure is coherent, if it is convex and positive homogeneous.
	\end{definition}
	
    Now, we are in the position to generalize the adjusted ES. 
 
	\begin{definition}\label{123}
		Let $\Rho = \{\rho_p\}_{p\in[0,1]}$ be a family of risk functionals and $g\in\mathcal{G}$. The map $\adjustedRiskMeasure:\mathcal{X}\rightarrow[-\infty,\infty]$ is defined for every $X \in \mathcal{X}$ by
		$$\adjustedRiskMeasure(X) \coloneqq \sup_{p\in [0,1]}\{\rho_p(X)-g(p)\}.$$

        The corresponding acceptance set is defined by $\mathcal{A}_{\Rho,g}\coloneqq \left\{X\in \mathcal{X}\,|\,\adjustedRiskMeasure(X)\leq 0\right\}.$
	\end{definition}
	
	\begin{remark}
	    For every $X\in\mathcal{X}$ we obtain that  $\adjustedRiskMeasure(X) = \inf\{m\in \mathbb{R}\,|\,X-m\in \mathcal{A}_{\Rho,g}\}.$
	\end{remark}
	
	The construction of the map $\adjustedRiskMeasure$ is inspired by the definition of the adjusted ES and it allows us to be even more flexible in capturing the tail behavior of the financial position. Indeed, if we set $\{\rho_p\}_{p\in[0,1]} = \{\text{ES}_p\}_{p\in[0,1]}$, then we obtain the adjusted ES as a special case. Other meaningful examples for families of risk measures $\mathcal{P}$ are discussed in Sections~\ref{sec:lvar} and~\ref{sec:def_new_risk_measures}.
    
	In general, if the family of risk functionals $\Rho$ consists of monetary risk measures, then the map $\adjustedRiskMeasure$ is a monetary risk measure, as we show next.
	
	\begin{proposition}\label{prop:adjustedRiskMeasure}
		The map $\adjustedRiskMeasure$ from Definition \ref{123} is a risk functional. Furthermore, if $\Rho$ is a family of risk measures, then $\adjustedRiskMeasure$ is a risk measure.
	\end{proposition}
	\begin{proof}
        The properties (i) and (ii) in Definition \ref{def} follow immediately from the fact that every $\rho_p$ admits these properties. To prove the second claim, let $X\in \mathcal{X}$ and $\Rho$ be a family of risk measures. Then, we get that $$\adjustedRiskMeasure(X) = \sup_{p\in [0,1]}\{\rho_p(X)-g(p)\}\geq\rho_0(X)-g(0)> -\infty.$$
	\end{proof}

 \begin{definition}[Adjusted risk measure]
     If the map $\adjustedRiskMeasure$ from Definition \ref{123} is a monetary risk measure, then we call it an \textbf{\textit{adjusted risk measure}}.
 \end{definition}
	
    \subsection{Examples: LVaR and beyond}\label{sec:lvar}

    There exists already another example of an adjusted risk meausre in the literature, namely the Loss VaR (LVaR), as introduced in~\cite{LVaR}.
    \begin{definition}[LVaR]
        For an increasing and right-continuous function $\alpha : [0,\infty) \rightarrow (0,1]$, we set  
		
		\begin{equation*}
			\text{LVaR}_{\alpha}(Y) \coloneqq \sup_{u \geq 0}\{\text{VaR}_{\alpha (u)}(Y)-u\}.
		\end{equation*}
    \end{definition}

    The similarity between adjusted risk measures and the LVaR is more complex and the next result allows us to discuss this in more detail. It is a generalization of~\cite[Proposition 4]{LVaR} and the proof works analogously to the one of~~\cite[Proposition 4]{LVaR}.
    
    To formulate this result, we say that a family of risk functionals $\Rho=\{\rho_{p}\}_{p\in[0,1]}$ is ordered, if for any two elements $\rho_p,\rho_q \text{ with } p\leq q$ it holds that
	$\rho_p(X)\leq \rho_q(X) \text{ for all }  X \in \mathcal{X}$, i.e.,~the family $\Rho$ is ordered via the pointwise partial order for maps.
	\begin{proposition}\label{rep}
		Let $X\in \mathcal{X}$, $\Rho = \{\rho_p\}_{p\in[0,1]}$ be an ordered family of risk functionals and $g\in \mathcal{G}$ be a left-continuous function. Then it holds that
		$$\adjustedRiskMeasure(X) = \sup_{u\in \textup{Im}(g)}\left\{\rho_{g_+^{-1}(u)}(X)-u\right\},$$
        where $g_+^{-1}(u) =\sup\{p\in[0,1]\,|\,g(p)\leq u\}$.
	\end{proposition}

    \begin{proof}
		
		To prove the inequality ``$\geq$'', let $u\in \text{Im}(g)$. Set $p=g_+^{-1}(u)$. Then, by the left continuity of $g$, it holds that $g(p)\leq u$. Thus, by $ \Rho$ being ordered we obtain that
		$$ \rho_p(X)-g(p)=\rho_{g_+^{-1}(u)}(X)-g(p) \geq \rho_{g_+^{-1}(u)}(X) - u.$$

		To prove the inequality ``$\leq$'', assume an arbitrary $p\in[0,1]$. We start with the situation of $g(p)<\infty$. Set $u=g(p)$. Then, it holds that $g_+^{-1}(u)\geq p$. Thus, by $ \Rho$ being ordered we obtain that
		$$ \rho_p(X)-g(p)=\rho_p(X)-u \leq \rho_{g_+^{-1}(u)}(X) - u.$$  
        
        The inequality also holds for $g(p)=\infty$, due to the convention of $\infty-\infty \,\,\text{is set to}\, -\infty$.
	\end{proof}

	\begin{remark}\label{rem:LVaR_as_adjustedRM}
		If the map $p\mapsto \rho_p(X)$ is left-continuous, then we are able to replace $u\in \text{Im}(g)$ with $u\geq 0$ in the previous result.
		Then, the representation looks similar to the $\LVaR$, but the definitions are not the same, because the domain and image of the functions $g_+^{-1}:[c,\infty]\rightarrow[0,1]$ for some $c\in\mathbb{R}$ and $\alpha:[0,\infty)\rightarrow(0,1]$ do not agree. So, on the one side, not every right-continuous inverse $g_+^{-1}$ of a left-continuous function $g\in\mathcal{G}$ can be used as function $\alpha$ in the definition of the  $\LVaR$. But, on the other side, if the function $\alpha$ from an $\LVaR$ is given, then for the function $g\in\mathcal{G}$ defined by $($with convention $\infty\cdot 0 \,\,\text{is set to}\,\,0)$
        $$g(p)=\alpha_-^{-1}(p)\cdot\mathds{1}_{[a(0),\,\sup_{u \geq 0}\alpha(u)]}(p)+\infty \cdot\mathds{1}_{(\sup_{u \geq 0}\alpha(u),1]}(p),$$
		where $\alpha_-^{-1}(p) = \inf\{x\in[0,\infty)\,|\,\alpha(x)\geq p\}$, we get by~\cite[Proposition 4]{LVaR} that for the family of functionals chosen by $\Rho = \{\VaR_p\}_{p\in[0,1]}$ it holds that 
		$\adjustedRiskMeasure(X)=\sup_{p \in [0,1]}\{\VaR_p(X)-g(p)\}=\LVaR_{\alpha}(X),$
		because of $\sup_{p \in [0,\alpha(0))}$ $\{\VaR_{p}(X)\}\leq \sup_{p \in [\alpha(0),\sup_{u \geq 0}\alpha(u)]}\{\VaR_p(X)-g(p)\}=\LVaR_{\alpha}(X)$.
	\end{remark}

    \begin{remark}\label{-infty}
        Since ${\VaR}_0$ is not a risk measure (because its domain $\mathcal{X}$ contains random variables which are not essentially bounded from below), Proposition~\ref{prop:adjustedRiskMeasure} is not applicable for $\Rho = \{\VaR_p\}_{p\in[0,1]}$. 
		But, note that $s=\sup\{p \in [0,1]\,|\,g(p)<\infty\}>0$. Let $u\in(0,s)$.  Then,
		$$\adjustedRiskMeasure(Y) = \sup_{p\in [0,1]}\{\VaR_p(Y)-g(p)\}\geq \VaR_{u}(Y)-g(u)>-\infty,$$
		due to the finiteness of $\VaR_u$. So, $\adjustedRiskMeasure$ is an adjusted risk measure.
	\end{remark}

    \begin{remark}
        \begin{enumerate}
            \item[(i)] In Proposition~\ref{rep} it is assumed that the family $\Rho = \{\rho_p\}_{p\in[0,1]}$ is ordered. This is a natural assumption which was used in the literature before, see e.g.,~\cite[Section 3.2]{pelve}. The families $\{\VaR_p\}_{p\in[0,1]}$ and $\{\ES_p\}_{p\in[0,1]}$ are ordered.
            \item[(ii)] In Remark~\ref{rem:LVaR_as_adjustedRM} we used the assumption that the map $p\mapsto \rho_p(X)$ is left-continuous. This is also satisfied for $\{\VaR_p\}_{p\in[0,1]}$ and $\{\ES_p\}_{p\in[0,1]}$. 
        \end{enumerate}
        The upcoming examples of families of risk measures satisfy both properties, see Examples~\ref{exam:ConditionalES},~\ref{exam:REVaR} and Section~\ref{sec:def_new_risk_measures}.  
    \end{remark}

    Now, we discuss different situations for the family of risk measures $\mathcal{P}$. The first one is motivated by the Solvency II regulation for insurers. The second and third are taken from the existing literature. 

    \begin{example}
        Solvency II requires that the future equity position $X$ of an insurance company satisfies $\VaR_{0.995}(-X)\leq 0$. Now, suppose that the insurer would like to improve the risk management by internally ensuring that the Expected Shortfall of $X$ at some level $p$, for instance $p=0.975$, is less than some constant $c<0$. This means $\ES_{0.975}(-X)\leq c$, i.e.,~the insurer would like to add a buffer for the internal model based on ES and solely would like to satisfy the legal requirements based on VaR. The total capital requirement of the insurer is then given as $\max\{\ES_{0.975}(-X)-c,\VaR_{0.995}(-X)\}$. This is the adjusted risk measure for $\mathcal{P} = \{\rho_p\}_{p\in[0,1]}$ with $\rho_p(X) =\ES_p(X)\cdot\mathds{1}_{[0,0.975]}(p) + \VaR_{p}(X)\cdot \mathds{1}_{(0.975,1]}(p)$ and $g(p) = c\cdot\mathds{1}_{[0,0.975]}(p) + \infty \cdot \mathds{1}_{(0.995,1]}(p)$, with the convention $\infty\cdot 0 \,\,\text{is set to}\,\, 0$.   
    \end{example}

    \begin{example}\label{exam:REVaR}
        Two examples of families of risk measures $\mathcal{P} = \{\rho_p\}_{p\in[0,1]}$ applied in the context of adjusted risk measures are the adjusted Rényi entropic VaR and the adjusted higher-order ES. For details on their definitions we refer to \cite{ZOU2023255} and \cite{ZOU20241}. For the adjusted Rényi entropic VaR, a family of entropic VaRs of some fixed order is used. This family is ordered and the map $p\mapsto \rho_p(X)$ by~\cite[Proposition 2.4]{ZOU2023255} continuous. For the adjusted higher-order ES, a family of higher-order ES is used. It is easy to see that this family is ordered. Furthermore, the map $p\mapsto \rho_p(X)$ is continuous, see~\cite[Lemma 3.7]{ZOU20241}.
    \end{example}    

    \begin{example}\label{exam:ConditionalES}
        As another example for the family $\mathcal{P} = \{\rho_p\}_{p\in[0,1]}$ we could use the conditional ES, given as $\rho_p(X) = \frac{1}{1-p}\int_{p}^1 \ES_q(X)\,\mathrm{d}q$ for $p\in[0,1)$ and $\rho_1(X) = \esssup X$. This family of risk measures is used for instance in~\cite[Section 3.2.1]{pelve}. This family $\mathcal{P}$ is ordered and for all $X$, the map $p\mapsto \rho_p(X)$ is left-continuous, because of integrating with respect to the Lebesgue-measure.
    \end{example}  
    
    So, there exist various examples of adjusted risk measures relevant for practice as well as of theoretical interest. 

    \subsection{Basic properties}

    As for monotonicity and cash additivity, we expect that if the underlying risk functionals have a certain property in common, then this property translates to the adjusted risk measure. This is true for convexity, star-shapedness and the following properties (ii)-(iv). If the target risk profile $g$ satisfies $g(0)=0$, then this is also true for the upcoming property (i).
    \begin{definition}
        We introduce the following properties of a risk measure $\rho$:
		\begin{enumerate}
            \item[(i)] Normalization: $\rho(0)=0,$
			\item[(ii)] Law-invariance: $\rho(X) = \rho(Y) \text{ for all }X,Y\in \mathcal{X}\text{ such that } X\sim Y,$
			\item[(iii)] Surplus invariance: $\rho(X) = \rho(\text{max}\{X,0\}) \text{ for all } X\in \mathcal{X} \text{ such that } \rho(X)\geq0 .$
		\end{enumerate}
        For $\mathcal{X}=L^1$, we introduce the following property:
        \begin{enumerate}
			\item[(iv)] Consistency with SSD: $\rho(X)\leq \rho(Y) \text{ for all }X,Y\in L^1\text{ such that } X\geq_{\text{SSD}}Y.$
        \end{enumerate}
    \end{definition}
    
    As mentioned, a bunch of properties translate from the underlying risk functionals to the adjusted risk measure. But, these situations are only sufficient to obtain the corresponding property of the map $\adjustedRiskMeasure$. So, the question arises of whether it is possible to obtain a certain property, but one of the underlying risk functionals does not satisfy this property. Example~\ref{ex123} shows that this is indeed possible. It uses a family of risk measures which is in general not convex, namely a family that contains $\VaR$ and ES. The example is based on step functions as target risk profiles. This implies that the ES values dominate the $\VaR$ values which results in a representation for the adjusted risk measure that is completely characterized by the ES values. In the example, we use the convention that $\infty\cdot 0 \,\,\text{is set to}\,\,0$.
	
	\begin{example}\label{ex123}
        Let $n\in\mathbb{N}$, $0=r_1<\dots<r_{n}<r_{n+1}=\infty$ and $0=p_0<p_1<\dots<p_n<p_{n+1}=1$ be arbitrary. We set $g\in\mathcal{G}$  as
        $$g(p)=r_1\cdot\mathds{1}_{[0,p_1]}(p)+\sum_{i=2}^{n+1}r_i\cdot\mathds{1}_{(p_{i-1},p_i]}(p).$$ Such a function is called a \textit{\textbf{step function}}. 
        
        Now, let $X\in L^1$, then we choose $\{\rho_p\}_{p\in[0,1]}$ as follows
        $$\rho_p(X)= \text{VaR}_p(X)\cdot\mathds{1}_{(p_1,p_2)}(p) + \text{ES}_p(X)\cdot\mathds{1}_{[0,1]\setminus (p_1,p_2)}(p)$$
		We obtain from 
		$$\sup_{p \in (p_1,p_2)}\{\text{VaR}_p(X)-g(p)\}=\sup_{p \in (p_1,p_2)}\{\text{VaR}_p(X)-r_2\}\leq\text{VaR}_{p_2}(X)-r_2\leq\text{ES}_{p_2}(X)-r_2$$
		that the adjusted risk measure is given by
		$\adjustedRiskMeasure(X)=\max_{i\in\{1,...,n\}}\{\text{ES}_{p_i}(X)-r_i\}$. Thus, we get that $\adjustedRiskMeasure$ is convex, but there are risk functionals in $\Rho$ who are not, namely all $\VaR$s. So the assumptions on the family of risk functionals are only sufficient, but not necessary.
	\end{example}

    \begin{remark}
        Another property related to our new concept is discussed in \cite{ptail}. This is the so-called $p$-tail property, which states that for a given value $p\in (0,1)$, a risk measure $\rho$ is a $p$-tail risk measure if for all $X,Y\in \mathcal{X}$ with $\text{VaR}_q(X) = \text{VaR}_q(Y)$ for every $q\in[p,1)$ it holds that $\rho(X)=\rho(Y)$. Note that every $p$-tail risk measure is law-invariant. Since an adjusted risk measure does not need to be law-invariant, it is not automatically a $p$-tail risk measure. However, if we assume that the family of risk functionals $\Rho=\{\rho_{s}\}_{s\in[0,1]}$ is ordered and the map $p\mapsto \rho_p(X)$ is left-continuous, then $\adjustedRiskMeasure$ is a $p$-tail risk measure when $g$ is constant on $(0,p)$ and $\rho_s$ has the p-tail property for $s\in[p,1)$.
    \end{remark}

    \subsection{Positive homogeneity and subadditivity}
	The analysis of subadditivity and positive homogeneity of an adjusted risk measure is more complex. To obtain equivalent properties, we use the following weak assumption, where $\mathcal{G}_0$ denotes the set of all functions $g\in\mathcal{G}$ with~$g(0) = 0$.
	\begin{assumption}\label{ass}
		Let $\Rho = \{\rho_p\}_{p\in[0,1]}$ be a family of risk functionals. Then the following should hold for all functions $g\in\mathcal{G}_0$ with $g(p)\in(0,\infty)$ for at least two $p\in(0,1]$:
		    There exist $\lambda>1$, $\varepsilon>0$ and $X\in\mathcal{X}$ with $q= \sup\{p\in[0,1]\,|\,g(p)=0\}$ such that 
      $$\sup_{p \in [0,1]}\{\rho_p(X)-g(p)\}=\sup_{p \in [q+\varepsilon,1]}\{\rho_p(X)-g(p)\}>-\infty\quad \text{and} \quad \sup_{p \in (q,1]}\left\{\rho_p(X)-\frac{g(p)}{\lambda}\right\}<\infty.$$
	\end{assumption}

    \begin{remark}
        Assumption~\ref{ass} implies that we can find a loss $X$ for which the supremum can be restricted to levels for which the target risk profile is strictly greater than $0$. Further, we obtain finiteness for the supremum with the target risk profile scaled by $\frac{1}{\lambda}$. In other words, Assumption~\ref{ass} guarantees that there exists a non-trivial case for the adjusted risk measure, where the target risk profile becomes relevant for the calculation of the adjusted risk measure.
    \end{remark}
	
    We can formulate an even more intuitive and stronger condition than Assumption~\ref{ass}, which is still satisfied by all positive homogeneous examples of families of risk functionals that we are using later on. This stronger condition is the cornerstone of the next proposition.

    \begin{proposition}\label{propalt}
        Let the family of risk functionals $\Rho$ be ordered and each risk functional is positive homogeneous. If it exists $X\in \mathcal{X}$ satisfying the following properties:
        \begin{enumerate}[(i)]
            \item $\rho_{p_1}(X)<\rho_{p_2}(X)$ for all $p_1,p_2\in[0,1]$ with $p_1<p_2$,
            \item $\rho_0(X)>-\infty$ and $\rho_1(X)<\infty,$
        \end{enumerate}
        then Assumption~\ref{ass} holds.
    \end{proposition}
    \begin{proof} 
    Let $g\in\mathcal{G}_0$ with $g(p_1),g(p_2)\in(0,\infty)$ for $p_1,p_2\in(0,1]$ with $p_1 < p_2$. Now, let $q$ be as in Assumption~\ref{ass}. Note that $q<p_2$. Then, for every $\alpha>0$ and $\lambda>1$ we get by (ii) that
    $$\sup_{p\in[0,1]}\{\rho_p(\alpha X)-g(p)\}\geq \rho_0(\alpha X) = \alpha\rho_0(X)>-\infty $$ and 
    $$\sup_{p\in[0,1]}\left\{\rho_p(\alpha X)-\frac{g(p)}{\lambda}\right\}\leq \alpha\rho_1( X)<\infty.$$
        Now, we conclude that
        \begin{align*}
            \sup_{p \in [0,1]}\{\rho_p(\alpha X)-g(p)\}&= \alpha \sup_{p \in [0,1]}\left\{\rho_p(X)-\frac{g(p)}{\alpha}\right\}\\
            &= \alpha \max\left\{\lim_{p\uparrow q}\rho_p(X),   \rho_q(X)-\frac{g(q)}{\alpha},\sup_{p \in (q,1]}\left\{\rho_p(X)-\frac{g(p)}{\alpha}\right\}\right\},
        \end{align*}
        where we used property (i) in the second equation.  From this, we conclude that for every $\alpha>\frac{g(p_2)}{\rho_{p_2}(X)-\rho_q(X)}$ it holds that 
        \begin{align}\label{eq:proofSufficientPositiveHomgeneous}
            \sup_{p \in [0,1]}\{\rho_p(\alpha X)-g(p)\}
            =\alpha \sup_{p \in (q,1]}\left\{\rho_p(X)-\frac{g(p)}{\alpha}\right\}. 
        \end{align}
        Further, for $\alpha$ large enough, it holds that 
        \begin{align*}
            \lim_{p\downarrow q}\left(\rho_p(X)-\frac{g(p)}{\alpha}\right)\leq \lim_{p\downarrow q}\rho_p(X) < \rho_{p_2}(X)-\frac{g(p_2)}{\alpha} \leq \sup_{p \in [p_2,1]}\left\{\rho_p(X)-\frac{g(p)}{\alpha}\right\},
        \end{align*}
        where we used (i) to obtain the second inequality. Together with~(\ref{eq:proofSufficientPositiveHomgeneous}), it follows that 
        \begin{align*}
            \sup_{p \in [0,1]}\{\rho_p(\alpha X)-g(p)\}
            = \sup_{p \in [p_2,1]}\left\{\rho_p(\alpha X)-g(p)\right\}.
        \end{align*}
      
        Concluding, Assumption~\ref{ass} is satisfied for $Y=\alpha X.$
    \end{proof}

    \begin{remark}
        The families of risk functionals $\{\ES_p\}_{p\in[0,1]}$ and $\{\VaR_p\}_{p\in[0,1]}$ satisfy the conditions in Proposition~\ref{propalt} for every  essentially bounded continuous random variable $X$.
    \end{remark}
    
	Now we are in the position to state our first result regarding the positive homogeneity and subadditivity. It is an auxiliary result, which gives us sufficient conditions to obtain the aforementioned properties.
 
    \begin{lemma}\label{lem}
		Let $\Rho = \{\rho_p\}_{p\in[0,1]}$ be a family of normalized risk functionals and assume that the map $p\mapsto\rho_p(X)$ be left-continuous for every $X\in\mathcal{X}$. Then, for every $g\in\mathcal{G}_0$ with $g(p)\in (0,\infty)$ for at most one $p\in(0,1]$ it holds:
		\begin{itemize}
			\item[\textup{(i)}] $\adjustedRiskMeasure$ is positive homogeneous if $\rho_p$ is positive homogeneous for all $p\in[0,1]$, 
			\item[\textup{(ii)}] $\adjustedRiskMeasure$ is subadditive if $\rho_p$ is subadditive for all $p\in[0,1]$.
		\end{itemize}
	\end{lemma}
	\begin{proof}
		We set $q= \sup\{p\in[0,1]\,|\,g(p)=0\}$.
		\begin{itemize}
		\item[(i)] We have to prove that $\text{epi}(\adjustedRiskMeasure)$ is a cone. Note that $\adjustedRiskMeasure$ is normalized. Hence, for $\lambda=0$ we obtain that $\adjustedRiskMeasure(0)\leq 0$. For $\lambda> 0$ and $(X,\alpha)\in \text{epi}(\adjustedRiskMeasure)$, we note that $(X,\alpha) \in \text{epi}(\rho_p)$ for all $p\in[0,q)$. Then, we get
        $$\adjustedRiskMeasure(\lambda X)= \lambda\sup_{p\in [0,q]}\left\{\rho_p(X)-\frac{g(p)}{\lambda}\right\}= \lambda\sup_{p\in [0,q)}\{\rho_p( X)\}\leq \lambda\sup_{p\in [0,q)}\{\alpha\}=\lambda \alpha,$$
		where we used the positive homogeneity of $\rho_p$ in the first equation and the left-continuity of $p\mapsto\rho_{p}(X)$ in the second equation.\\
		\item[(ii)] Let $(X,\alpha),(Y,\beta)\in \text{epi}(\adjustedRiskMeasure)$. From $(X+Y,\alpha +\beta) \in \text{epi}(\rho_p)$ for $p\in[0,q)$ we obtain that
			$$\adjustedRiskMeasure(X+Y)= \sup_{p\in [0,q)}\{\rho_p(X+Y)\}\leq \sup_{p\in [0,q)}\{\alpha+\beta\}=\alpha +\beta$$
		where we used the left-continuity of $p\mapsto\rho_{p}(X)$ in the first equation analogously as in $(\text{i})$ and the subadditivity of $\rho_p$ in the first inequality.
		\end{itemize}
	\end{proof}
	
	Now, under Assumption~\ref{ass} we can even state equivalent conditions to positive homogeneity and subadditivity of an adjusted risk measure. This analysis is inspired by~\cite[Proposition 3.2]{AES}, which says that positive homogeneity of the adjusted ES is equivalent to the target profile being of a specific form. Our result generalizes this observation and~\cite[Proposition 3.2]{AES} is a special case of it.

    \begin{theorem}\label{theo}
        Let $\Rho = \{\rho_p\}_{p\in[0,1]}$ be a family of normalized risk functionals and $g\in\mathcal{G}$. If we assume that for all $X\in\mathcal{X}$ the map $p\mapsto\rho_p(X)$ is left-continuous and Assumption~\ref{ass} holds, then        
		\begin{itemize}
		\item[\textup{(i)}] If for all $p\in[0,1]$ the risk functional $\rho_p$ is positively homogeneous then: 
        
        $\adjustedRiskMeasure$ is positively homogeneous if and only if $g(0)=0$ and $g(p)\in (0,\infty)$ for at most one $p\in(0,1].$
		\item[\textup{(ii)}] If for all $p\in[0,1]$, the risk functional $\rho_p$ is star-shaped and subadditive, $\Rho$ is ordered and $g(0)\leq 0$, then:
  
		$\adjustedRiskMeasure$ is subadditive if and only if $g(0)=0$ and $g(p)\in (0,\infty)$ for at most one $p\in(0,1].$ 
		\end{itemize}
	\end{theorem}
	\begin{proof}
		Set $q= \sup\{p\in[0,1]\,\,|\,\,g(p)=0\}$ and choose an arbitrary $X\in\mathcal{X}$.

        \begin{itemize}
		      \item[\textup{(i)}] Lemma \ref{lem} (i) gives us the implication ``$\Leftarrow$''. For the proof of the implication ``$\Rightarrow$'', note that positive homogeneity of $\adjustedRiskMeasure$ and normalization of the risk functionals imply for all $\lambda > 0$ that
              \begin{align*}
                  -g(0) = \lambda \adjustedRiskMeasure(0) = -\lambda g(0),
              \end{align*}
              which is only possible if $g(0)=0$. Now, we argue by contraposition. To do so, assume that there exist two points, for which $g$ attains values in $(0,\infty)$. Due to Assumption~\ref{ass}, there exists $\varepsilon>0$ and a sequence $\{p_n\}_{n\in\mathbb{N}}\subseteq [q+\varepsilon,1]$ such that 
              \begin{align*}
                    \adjustedRiskMeasure(X) = \lim_{n\rightarrow\infty}(\rho_{p_n}(X) -g(p_n)).
              \end{align*}
        Note that Assumption~\ref{ass} implies that $|\adjustedRiskMeasure(X)|<\infty$. Therefore, $\rho_{p_n}(X) -g(p_n)\in\mathbb{R}$ for almost every $n$. this gives us that for every $\lambda>1$ it also holds that $\rho_{p_n}(X),g(p_n)\in\mathbb{R}$ for almost every $n$. Together with the positive homogeneity of $\rho_{p}$ for all $p\in[0,1]$, we get that for every $\lambda>1$ it holds $\rho_{p_n}(\lambda X)= \lambda \rho_{p_n}(X)$ for almost every $n$. With the help of this property, we have for every $\lambda>1$ that
        \begin{align*}
            \adjustedRiskMeasure(\lambda X)&\geq\sup_{n \in \mathbb{N}}\left\{\rho_{p_n}(\lambda X)-g(p_n)\right\}\\ 
            &=\lambda \sup_{n \in \mathbb{N}}\left(\rho_{p_n}(X) -\frac{g(p_n)}{\lambda}\right)\\
            &>\lambda \sup_{n \in \mathbb{N}}\left(\rho_{p_n}(X) -g(p_n)\right)\\
            &\geq \lambda \lim_{n \mapsto \infty}\left(\rho_{p_n}(X) -g(p_n)\right)\\
            &= \lambda \adjustedRiskMeasure(X).
        \end{align*} 
        
		So, $(X,\adjustedRiskMeasure(X))\in \text{epi}(\adjustedRiskMeasure)$, but for $\lambda>1$ it holds that $(\lambda X,\lambda \adjustedRiskMeasure(X))\notin \text{epi}(\adjustedRiskMeasure)$. Hence, $\text{epi}(\adjustedRiskMeasure)$ is not a cone, i.e.,~$\adjustedRiskMeasure$ is not positively homogeneous. 
		
		\item[\textup{(ii)}]Lemma \ref{lem} (ii) gives us the implication ``$\Leftarrow$''. To prove the implication ``$\Rightarrow$'', note first that an analogous argumentation as in part (i) shows that $g(0)\geq 0$, which implies that $g(0) = 0$. Now, we again argue by contraposition.

        Due to Assumption~\ref{ass}, there exist $\varepsilon>0$ and a sequence $\{p_n\}_{n\in\mathbb{N}}\subseteq [q+\varepsilon,1]$ as in the proof of (i). By Assumption~\ref{ass}, we have that  $$\lim\limits_{n\rightarrow\infty}\left(\rho_{p_n}(X) -g(p_n)\right)\geq \rho_q(X).$$ Assume that equality holds. Then,  
        $$\lim\limits_{n\rightarrow\infty}\rho_{p_n}(X) - \lim_{n\rightarrow\infty}g(p_n)=\lim\limits_{n\rightarrow\infty}\left(\rho_{p_n}(X) -g(p_n)\right)=\rho_q(X).$$ 
        
        Note, Assumption~\ref{ass} gives us that $\adjustedRiskMeasure(X)\in\mathbb{R}$. Hence, it has to hold that $\lim\limits_{n\rightarrow\infty}\rho_{p_n}(X)\in\mathbb{R}$ and $\lim_{n\rightarrow\infty}g(p_n)\in\mathbb{R}$. This allows us to conclude that        $$0<\lim_{n\rightarrow\infty}g(p_n)=\rho_q(X)-\lim_{n\rightarrow\infty}\rho_{p_n}(X).$$
        So, $\lim\limits_{n\rightarrow\infty}\rho_{p_n}(X)<\rho_q(X)$, which is a contradiction to $\mathcal{P}$ being ordered, i.e.,~it holds that
        \begin{align*}
            \lim\limits_{n\rightarrow\infty}\left(\rho_{p_n}(X) -g(p_n)\right)>\rho_q(X) = \sup_{p\in [0,q]}\{\rho_p(X)\}.
        \end{align*}

        Further, note that
        $$\rho_q(X)=\sup_{p\in[0,q]}\{\rho_p(X)\}=\inf_{n\in\mathbb{N}}\{\adjustedRiskMeasureInfimal(X)\}.$$
        Now choose $n\in\mathbb{N}$ with $\adjustedRiskMeasureInfimal(X)<\adjustedRiskMeasure(X)$, then with the star-shapedness of $\rho_p$ we obtain:
        \begin{align*}
            n\adjustedRiskMeasure\left(\frac{X}{n}\right)&=n\sup_{p \in [0,1]}\left\{\rho_{p}\left(\frac{X}{n}\right)-g(p)\right\}\\
            &\leq n\frac{1}{n}\sup_{p \in [0,1]}\left\{\rho_{p}(X)-ng(p)\right\}\\
            &= \adjustedRiskMeasureInfimal(X)\\
            &<\adjustedRiskMeasure(X),
        \end{align*}
        which contradicts subadditivity.
        \end{itemize}
	\end{proof}

 \subsection{Finiteness and continuity}\label{sec:finitenessGeneral}
     
     Now we discuss conditions under which an adjusted risk measure is finite and continuous. Our results are tailor-made for the concrete risk measures that we introduce in the upcoming Sections~\ref{sec:def_new_risk_measures} and~\ref{case}. To guarantee finiteness, we restrict to special kinds of target risk profiles. The following example shows that a general result is not possible.

     \begin{example}
        Assume the adjusted ES with a target risk profile $g$ such that $g(1)<\infty$. Recall that $\infty-\infty$ is set to $-\infty$. Then, for $X\in L^1$ such that $\esssup X = \infty$, we obtain $\text{ES}^g(X) = \esssup X - g(1) = \infty$. Furthermore, $\text{ES}^g(X) = \infty$ is also possible even if $g(1)=\infty$. Let $X,Z\in L^1$ be independent and normal distributed with mean $0$ and standard deviations $\sigma_X,\sigma_Z>0$ such that $\sigma_X>\sigma_Y$. Then, set $g(p) = \text{ES}_p(Z)$. Then, it holds that 
        $$\sup_{p\in[0,1]}\{\text{ES}_p(X)-\text{ES}_p(Z)\}
        =\lim_{p\rightarrow 1}\left\{(\sigma_X-\sigma_Z)\frac{\varphi(\Phi^{-1}(p))}{1-p}\right\}=\infty,$$
        where $\varphi$, respectively $\Phi$, denotes the PDF, respectively the CDF, of a standard normally distributed random variable. 
     \end{example}

      Motivated by the previous example, we introduce the following assumption, which is crucial to obtain finiteness and continuity of an adjusted risk measure.
      \begin{assumption}\label{assump:finiteness}
          For $g\in\mathcal{G}_0$ it should hold that $g$ is lower semicontinuous and that 
          \begin{align*}
              0<p_1<p_2<1,
          \end{align*}
          where $p_1 = \max\{p\in[0,1]\,|\,g(p)=0\}$ and $p_2=\max\{p\in[0,1]\,|\,g(p)<\infty\}$.
      \end{assumption}

      Now we are in the position to state sufficient conditions to obtain finiteness. 
      \begin{lemma}\label{lem:finiteness}
          Let $\mathcal{X} = L^1$ and $g\in \mathcal{G}_0$ satisfies Assumption~\ref{assump:finiteness}. Let $p_1$ and $p_2$ be given as in Assumption~\ref{assump:finiteness}. Let $\Rho = \{\rho_p\}_{p\in[0,1]}$ be an ordered family of risk functionals such that  $\rho_p$ is finite-valued for all $p\in[p_1,p_2]$. Then, $\adjustedRiskMeasure$ is finite-valued. 
    \end{lemma}

    \begin{proof}
        For $X\in L^1$ it holds that $-\infty<\rho_{p_1}(X)\leq 
            \adjustedRiskMeasure(X)\leq \rho_{p_2}(X)<\infty$. 
    \end{proof}
    
    To obtain continuity we use the next result, in which part (b) is a version of~\cite[Corollary 3.14]{Farkas}. To do so, for a subset $\mathcal{A}$ of $L^1$, we denote by $\interior(\mathcal{A})$ the interior of $\mathcal{A}$ with respect to the $L^1$-norm.
    \begin{proposition}\label{prop:continuity}
        Let $\mathcal{X} = L^1$ and $g\in \mathcal{G}_0$ satisfies Assumption~\ref{assump:finiteness}. Let $p_1$ and $p_2$ be given as in Assumption~\ref{assump:finiteness}. Let $\Rho = \{\rho_p\}_{p\in[0,1]}$ be a family of convex risk functionals. Then, the following statements hold:
        \begin{enumerate}
            \item[(a)] If $[p_1,p_2]\rightarrow[-\infty,\infty],p\mapsto \rho_p(X)$ is upper semicontinuous for all $X\in L^1$ and \begin{align}\label{eq:condition_continuity}
            \bigcap\limits_{p\in[p_1,p_2]}\interior(\{X\in L^1\,|\,\rho_p(X)\leq 0\})\neq \emptyset,
        \end{align}
        then $\interior({\mathcal{A}_{\Rho,g}})\neq \emptyset$.
            \item[(b)] If $\interior({\mathcal{A}_{\Rho,g}})\neq \emptyset$, then $\adjustedRiskMeasure$ is finite-valued and continuous.
        \end{enumerate}
    \end{proposition}

    \begin{proof} 
        We first prove (a). To do so, note that for every $p\in[p_1,p_2]$, we can apply~\cite[Corollary 3.14]{Farkas} to obtain that the convexity of $\rho_p$ together with~(\ref{eq:condition_continuity}) imply that $\rho_p$ is continuous. Then, by the fact that the supremum of a family of lower semicontinuous functions is lower semicontinuous, see~\cite[Lemma 2.41]{aliprantis_infinite_2006}, we get that $\adjustedRiskMeasure$ is lower semicontinuous. Hence, by~\cite[Lemma 2.5 and Remark 2.6]{Farkas} we obtain that
        \begin{align*}
            \interior({\mathcal{A}_{\Rho,g}})=\left\{X\in L^1\,\middle|\, \sup_{p\in[p_1,p_2]}\{\rho_p(X)- g(p)\}<0\right\}.
        \end{align*}
        The upper semicontinuity of the map $[p_1,p_2]\rightarrow[-\infty,\infty],p\mapsto \rho_p(X)-g(p)$ for all $X\in L^1$ gives us the following:
        \begin{align*}
            \interior({\mathcal{A}_{\Rho,g}})&=\left\{X\in L^1\,\middle|\,\forall p\in [p_1,p_2]:\rho_p(X)-g(p)<0\right\}\\
            &=\bigcap_{p\in[p_1,p_2]}\{X\in L^1\,|\,\rho_p(X)-g(p)<0\}\\
            &= \bigcap_{p\in[p_1,p_2]}\interior(\{X\in L^1\,|\,\rho_p(X)-g(p)\leq 0\}).
        \end{align*}
        Together with~(\ref{eq:condition_continuity}) we obtain that $\interior(\mathcal{A}_{\Rho,g})\neq\emptyset$. 
        
        Note that $\adjustedRiskMeasure$ is convex (compare with the discussion at the end of Section~\ref{sec:def_adjustedRiskMeasure}). Then, the non-empty interior of the acceptance set $\mathcal{A}_{\Rho,g}$ and the convexity of $\adjustedRiskMeasure$ implies by~\cite[Corollary 3.14]{Farkas} that $\adjustedRiskMeasure$ is finite-valued and continuous, which also shows part (b).
    \end{proof}

    \begin{remark}
        The reason to assume~(\ref{eq:condition_continuity}) in Proposition~\ref{prop:continuity} is twofold. First, it guarantees that every $\rho_p$ is finite-valued and continuous. Second, it ensures that the interiors of the acceptance sets of the risk measures from the family $\{\rho_p\}_{p\in[0,1]}$ have one element in common and the proof shows that this is sufficient to obtain a non-empty interior of the acceptance set of $\adjustedRiskMeasure$. The latter allows for the application of~\cite[Corollary 3.14]{Farkas}.
    \end{remark}

    \begin{example}
        We check the conditions of Proposition~\ref{prop:continuity} in case of the adjusted ES for $g\in\mathcal{G}_0$ already satisfying Assumption~\ref{assump:finiteness}. For all $X\in L^1$, the function $[p_1,p_2]\rightarrow\mathbb{R},p\mapsto \text{ES}_p(X)$ is upper semicontinuous, indeed even continuous, see e.g.,~the proof of~\cite[Proposition 3.2]{AES}. Further, it is well-known that the Expected Shortfall is convex. By cash-additivity of the Expected Shortfall, we obtain for every $p\in(0,1)$ that $\text{ES}_p(-1) = -1 <0$, which means that $-1$ is an element of the interior of the ES acceptance set at level $p$. Hence, all conditions of Proposition~\ref{prop:continuity} are satisfied and therefore, the adjusted ES is finite-valued and continuous.
    \end{example}
	
	\section{Examples of adjusted risk measures}\label{sec:def_new_risk_measures}

    Now, we introduce concrete adjusted risk measures. This allows us to compare their performances to the adjusted ES in the case study in Section~\ref{case} later on. Also, we discuss algebraic properties of these new risk measures and write down dual representations for them.

    From now on, we always restrict ourselves to target risk profiles from the set $\mathcal{G}_0$.
    
    \subsection{Definitions}
 
    As a prerequisite, we introduce two known monetary risk measures, namely Range-Value-at-Risk (RVaR), see~\cite[Definition 3.2.3]{Robust}, and expectiles, see~\cite[Equation (5)]{expectiles}.
	
	\begin{definition}(RVaR and expectiles)\label{def RM}
        Let $X\in L^1$. Then for $0\leq\alpha_1<\alpha_2\leq1$, we define the RVaR at levels $\alpha_1$ and $\alpha_2$ of $X$ by
			$$\text{RVaR}_{\alpha_1,\alpha_2}(X) \coloneqq \frac{1}{\alpha_2-\alpha_1}\int_{\alpha_1}^{\alpha_2}\text{VaR}_u(X)\,\mathrm{d}u.$$
        Further, we define $\text{RVaR}_{\alpha,\alpha}(X) := \VaR_{\alpha}(X)$ for all $\alpha\in[0,1]$.

        The expectile at level $q\in(0,1)$ of $X$ is denoted by $e_q(X)$ and defined as the solution of the equation
			$$qE[\max\{X-e_q(X),0\}]=(1-q)E[-\min\{X-e_q(X),0\}].$$
        Further, we set $e_0(X) = \essinf X$ and $e_1(X) = \esssup X$.
	\end{definition}

	\begin{remark}
            For further details and a recent overview on RVaR, we refer to \cite{RVaR2,elictable,Robust}. For expectiles we refer to~\cite{1234,expectiles}.
	\end{remark}
 
	The first adjusted risk measure that we construct follows the idea to use the VaR up to a certain level $r\in(0,1)$ and ES otherwise. This idea comes from the following observation: By using the adjusted ES in combination with a step function as target risk profile, we have to calculate a maximum of affine transformations of ES values. But, all these ES values are based on the tail of the underlying distribution. So, the whole tail is evaluated multiple times. In our case, in combination with step functions, we can force that only one ES is used and VaRs otherwise.
	
	\begin{definition}(Simplified composed risk measure)\label{def SCRM}
		Let $X\in L^1$ and $g\in\mathcal{G}_0$ and $\Rho = \{\rho_p\}_{p\in[0,1]}$ defined for some $r\in (0,1)$ as follows: $$\rho_p(X)= \text{VaR}_p(X)\mathds{1}_{[0,r]}(p) + \text{ES}_p(X)\mathds{1}_{(r,1]}(p)$$
		  In this case, we call $\adjustedRiskMeasure$ a \textbf{\textit{simplified composed risk measure (SCRM)}} and for clarification, we also write $\scrm = \adjustedRiskMeasure$.
	\end{definition}

	Next, we follow a similar approach as for the SCRM, but we replace the VaR with the RVaR. This approach allows to consider the whole tail, but we avoid that the involved risk measures are based on overlapping regions of the tail of the distribution.
	
	\begin{definition}(Composed risk measure)\label{CRM}
		Let $X\in L^1$ and $g\in\mathcal{G}_0$. For a set of levels $\mathcal{L} = \{p_k\}_{k\in\{0,1,\dots,n,n+1\}}\subset[0,1]$ with $n\in\mathbb{N}$ and $0=p_0<p_1<\dots<p_{n}<p_{n+1}=1$, let $\Rho=\{\rho_p\}_{p\in[0,1]}$ be given by
        $$\rho_p(X)=\text{RVaR}_{p,p_1}(X)\mathds{1}_{[0,p_1]}(p)+\sum_{i=2}^{n}\text{RVaR}_{p,p_i}(X)\mathds{1}_{(p_{i-1},p_i]}(p)+\text{ES}_p(X)\mathds{1}_{(p_{n},1]}(p).$$
        In this case, we call $\adjustedRiskMeasure$ a \textbf{\textit{composed risk measure (CRM)}} and for clarity, we also write $\crm = \adjustedRiskMeasure$.
	\end{definition}

     \begin{remark}
         Note that $\RVaR_{0,0}(X) = \VaR_{0}(X)$ is not possible and therefore, the family of risk functionals in the previous definition is a family of risk measures.    
     \end{remark}
	
	If we use the CRM together with a step function as in Example~\ref{ex123}, then the CRM reduces to a maximum over VaR and ES values. This effect can be avoided if we only work with a finite number of RVaRs and fix the first level. This is done in the following definition.
 
    \begin{definition}(Fixed composed risk measure)\label{FCRM}
		Based on the same setup as in Definition~\ref{CRM}, let $\Rho=\{\rho_p\}_{p\in[0,1]}$ be given by
        $$\rho_p(X)=\text{RVaR}_{0,p_1}(X)\mathds{1}_{[0,p_1]}(p)+\sum_{i=2}^{n}\text{RVaR}_{p_{i-1},p_i}(X)\mathds{1}_{(p_{i-1},p_i]}(p)+\text{ES}_{p}(X)\mathds{1}_{(p_{n},1]}(p).$$
        In this case, we call $\adjustedRiskMeasure$ a \textbf{\textit{fixed composed risk measure (FCRM)}} and for clarity, we also write $\fcrm = \adjustedRiskMeasure$.
	\end{definition}
 
    \begin{remark}
        In this section, we focus on the (S)CRM. In the case study in the next section, we then also apply the FCRM. 
    \end{remark}	
	
	Next, we consider a construction that differs completely from the previous constructions, by using a family of expectiles.
	
	\begin{definition}(Adjusted expectile risk measure)\label{AERM}
		Let $\mathcal{X}=L^1$, $g\in\mathcal{G}_0$ and $\Rho = \{e_p\}_{p\in[0,1]}$. We call $\adjustedRiskMeasure$ an \textbf{\textit{adjusted expectile risk measure (AERM)}} and write $\aerm=\adjustedRiskMeasure$.
	\end{definition}
	
	To illustrate the new adjusted risk measures, we use step functions as target risk profiles:
	
	\begin{example}\label{example}
		Let $X\in L^1$ be arbitrary. Then, for the step function $g$ from Example~\ref{ex123}, the following representations hold:
		\begin{enumerate}[(i)]
			\item SCRM: $\scrm(X) =\max\left\{\max\limits_{i\in\{1,...,j\}}\{\text{VaR}_{p_i}(X)-r_i\},\max\limits_{i\in\{j+1,...,n\}}\{\text{ES}_{p_i}(X)-r_i\}\right\},$ where $j\in\{1,\dots,n\}$ such that $p_j\leq r<p_{j+1}$. 
			\item CRM and FCRM in case in which the levels of $g$ and $\Rho$ are the same:
				\begin{align*}
				    \crm(X)&=\max\left\{\max_{i\in\{1,...,n-1\}}\{\text{VaR}_{p_{i}}(X)-r_i\},\text{ES}_{p_n}(X)-r_n\right\},\\
                    \fcrm(X) &= \max\left\{\max_{i\in\{1,...,n-1\}}\{\text{RVaR}_{p_{i-1},p_{i}}(X)-r_i\},\text{ES}_{p_n}(X)-r_n\right\}.
				\end{align*}
			\item AERM: $\aerm(X) = \max\limits_{i\in\{1,...,n\}}\left\{e_{p_{i}}(X)-r_i\right\}.$
		\end{enumerate}
        The CRM and the SCRM are equal, if $r=p_{n}$. The previous representations can be compared with the adjusted ES by using the same step function $g$. In this case, it holds that $$\ES^{g}(X) = \max\limits_{i\in\{1,...,n\}}\{\text{ES}_{p_{i}}(X)-r_i\}.$$
	\end{example}

    \subsection{Standard properties}\label{Standart properties}
    
    Now, we discuss properties of our new risk measures. To do so, we first mention that the RVaR is a positive homogeneous, star-shaped and law-invariant risk measure, see~\cite[Proposition 1]{RVaR}. Expectiles are star-shaped and law-invariant. For levels $q\geq \frac{1}{2}$, expectiles are consistent with SSD, convex and therefore coherent, see e.g.,~\cite{1234,expectiles}.
	
	First, we discuss if the new risk measures satisfy Assumption~\ref{ass}. To do so, one can easily check that the conditions of Proposition~\ref{propalt} are satisfied in the case of SCRM, CRM and FCRM for a continuously and essentially bounded random variable $X$. Second, note that SCRM, CRM, FCRM and AERM are all normalized, star-shaped and law-invariant. Further properties are summarized in Table~\ref{table1}. For the sake of brevity, additional explanations regarding the properties in this table are shifted to Appendix~\ref{sec:explanationsTable}.
	
	\begin{table}[H]
		$$\begin{tabular}{|c||c|c|c|c|c|c|c|c|c|}
			\hline
			\textbf{Properties}&VaR&RVaR&ES&EP&adj. ES&LVaR&SCRM&CRM&AERM\\
			\hline
			\hline
			convex&X&X&\checkmark&\checkmark&\checkmark&X&X&X&X \\
			\hline
			positive homogeneous&\checkmark&\checkmark&\checkmark&X&X&X&X&X&X\\
			\hline					
			subadditive&X&X&\checkmark&X&X&X&X&X&X\\
			\hline
			consistent with SSD&X&X&\checkmark&\checkmark&\checkmark&X&X&X&X\\
			\hline
			surplus invariant&\checkmark&X&X&X&X&\checkmark&X&X&X\\
			\hline
		\end{tabular}$$
		\caption{\footnotesize{Properties of risk measures:
			is fulfilled $= \checkmark$, is not fulfilled = X, EP =  expectile, adj. ES = adjusted ES.}}
        \label{table1}
	\end{table} 

    We see that neither SCRM nor CRM satisfy any of the properties in Table~\ref{table1}. The intuition behind this is as follows: both can be seen as a mixture between the adjusted ES and the LVaR, which do not have property in common. The AERM does not satisfy a single property in Table~\ref{table1} in general. But, note that if the target risk profile is zero for all level smaller than $\frac{1}{2}$, then we obtain convexity and consistency with SSD, i.e.,~the same properties as for the adjusted ES. 

    \subsection{Finiteness and continuity}\label{sub}
    
    Now, we use the results from Section~\ref{sec:finitenessGeneral} to discuss finiteness and continuity for the new risk measures. To do so, we assume that the target risk profile $g$ satisfies Assumption~\ref{assump:finiteness}. Here, we set 
    $$a= \max\{p\in[0,1]\,|\,g(p)=0\}\,\, \text{and}\,\, b=\max\{p\in[0,1]\,|\,g(p)<\infty\}$$. 

    \textit{SCRM and CRM:} For $X\in L^1$ we obtain for the SCRM that $-\infty<\text{VaR}_{a}(X)\leq \rho_{r,g}^{\text{SC}}(X)\leq \text{ES}_{b}(X)<\infty.$
    In general, continuity cannot hold for the SCRM and CRM, because they can simplify to a single VaR, which is not continuous according to ~\cite[Proposition 4.2]{Farkas}.

    \textit{FCRM:} The FCRM is finite-valued and continuous. This follows by the fact that the FCRM is the maximum of a finite number of continuous functions, namely the maximum of a finite number of (affine transformations of) RVaRs and (if $p_n<b$) the map $X\mapsto \sup_{p\in(p_n,1]}\{\text{ES}_p(X)-g(p)\}$. Indeed, every RVaR is of the form $\text{RVaR}_{p,q}(X)=\frac{\text{ES}_p(X)-\text{ES}_q(X)}{q-p}$ for $X\in L^1$ according to~\cite{https://doi.org/10.1111/mafi.12270}[Example 4.6]. Hence, the continuity of the RVaR follows directly from the continuity of the ES. The continuity of the map $X\mapsto \sup_{p\in(p_n,1]}\{\text{ES}_p(X)-g(p)\}$ is obtained as follows: Set the function $h$ as $h(p)=g(p)$ for all $p\in(p_n,b]$ and $h(p_n)=\lim_{p\downarrow p_n}g(p)$.  Then it holds that
    $$\sup_{p\in(p_n,1]}\{\text{ES}_p(X)-g(p)\}=\sup_{p\in[p_n,b]}\{\text{ES}_p(X)-h(p)\}.$$ 
    Note, $h$ is lower semicontinuous. Together with the upper semicontinuity of the Expected Shortfall we obtain from Proposition~\ref{prop:continuity} the finiteness and the continuity of the FCRM.

    \textit{AERM:} Expectiles are convex for levels $q\geq\frac{1}{2}$, see e.g.,~\cite[Proposition 1]{BELLINI201441}. Hence, if we set $a=\frac{1}{2}$ then the AERM is convex. Note that every expectile is (lower semi)continuous by~\cite[Theorem 10]{BELLINI201441}\footnote{Note that convergence with respect to the $L^1$-norm implies convergence in the Wasserstein distance, see e.g.,~\cite[Section 3.2]{BELLINI201441}.}. Hence, the AERM is lower semicontinuous, as the supremum of a family of lower semicontinuous functions (\cite[Lemma 2.41]{aliprantis_infinite_2006}). Then, by~\cite[Lemma 2.5 and Remark 2.6]{Farkas}, we obtain that
    \begin{align*}
        \interior(\mathcal{A}_{\mathcal{P},g}) = \{X\in L^1\,|\,\adjustedRiskMeasure(X)<0\}.
    \end{align*}
    Now, the cash additivity of expectiles implies that 
    $\sup_{p\in[p_1,p_2]}\{e_p(-1)- g(p)\}=-1<0,$ which shows that $-1\in\interior(\mathcal{A}_{\mathcal{P},g})$. Hence, all conditions of Proposition~\ref{prop:continuity} (b) are satisfied, which implies that the AERM is finite-valued and continuous. 

    Concluding, we obtain that the SCRM, the (F)CRM and the AERM are finite-valued if the target risk profile $g$ satisfies Assumption~\ref{assump:finiteness}. Furthermore, for such target risk profiles, the FCRM and the AERM are continuous.

    \subsection{Dual representations}\label{sec:dual}

    This  subsection is motivated by~\cite[Proposition 3.7]{AES}, in which a dual representation for the adjusted ES is given. We provide analogous representations for SCRM, CRM and AERM. 
    
    We start with the case of AERM. For this, given a set $A$, we denote by  $\delta_A$ the indicator function for which we have that $\delta_A(x)=0$ if $x\in A$ and $\delta_A(x) = \infty$ otherwise. For the next result we use the conventions that $\frac{1}{0}\,\,\text{is set to}\,\,\infty$ and $\frac{1}{\infty}=0$.
    
    \begin{proposition}\label{prop:dual_expectiles}
        Let $X\in L^1$ and $g\in\mathcal{G}_0$ with $g(1)=\infty$. Then, it holds that 
        \begin{align*}
            \aerm(X) = \max\left\{\sup_{p\in \left[0,\frac{1}{2}\right]}\left\{ \inf_{\mathbb{Q}\in \mathcal{P}^\infty_{\mathbb{P}}}\{E_{\mathbb{Q}}[X]-g(p)+\delta_{[0,c(\mathbb{Q})]}(p)\}\right\},\sup_{\mathbb{Q}\in\mathcal{P}^\infty_{\mathbb{P}}}\{E_{\mathbb{Q}}[X]-g(1-c(\mathbb{Q}))\}\right\},
        \end{align*}
        where $\mathcal{P}^\infty_{\mathbb{P}} = \left\{\mathbb{Q}\in \mathcal{P}\,\middle|\,\mathbb{Q} \ll \mathbb{P},\frac{\mathrm{d}\mathbb{Q}}{\mathrm{d}\mathbb{P}}\in L^\infty,\frac{\mathrm{d}\mathbb{Q}}{\mathrm{d}\mathbb{P}}>0 \,\,\text{a.s.}\right\}$ and $c(\mathbb{Q})=\frac{\essinf\left(\frac{\mathrm{d}\mathbb{Q}}{\mathrm{d}\mathbb{P}}\right)}{\essinf\left(\frac{\mathrm{d}\mathbb{Q}}{\mathrm{d}\mathbb{P}}\right)\,+\,\esssup\left(\frac{\mathrm{d}\mathbb{Q}}{\mathrm{d}\mathbb{P}}\right)}$.
    \end{proposition}

    \begin{proof}
        By the assumptions on $g$, it holds that $\aerm(X)=\sup_{p\in[0,1)}\{e_p(X)-g(p)\}$. Then, \cite[Proposition 8]{BELLINI201441} gives us the following dual representation for an expectile:
        \begin{align*}
            e_q(X)=\min_{\mathbb{Q}\in \mathcal{M}_q}\{E_{\mathbb{Q}}[X]\} \mathds{1}_{\left[0,\frac{1}{2}\right)}(q)+ \max_{\mathbb{Q}\in \mathcal{M}_q}\{E_{\mathbb{Q}}[X]\} \mathds{1}_{\left[\frac{1}{2},1\right)}(q),
        \end{align*}
        where $\mathcal{M}_q = \left\{ \mathbb{Q}\in \mathcal{P}^\infty_{\mathbb{P}}\,\middle|\, \frac{\esssup\left(\frac{\mathrm{d}\mathbb{Q}}{\mathrm{d}\mathbb{P}}\right)}{\essinf\left(\frac{\mathrm{d}\mathbb{Q}}{\mathrm{d}\mathbb{P}}\right)}\leq \max\left\{\frac{q}{1-q},\frac{1-q}{q}\right\}\right\}\,\text{and}\,\,q\in(0,1)$. For $q=0$ it holds that $$e_0(X) =\essinf(X) =\inf_{\mathbb{Q}\in\mathcal{P}_{\mathbb{P}}^{\infty}}E_{\mathbb{Q}}[X],$$
        see e.g.,~\cite[Example 4.39]{follmer_stochastic_2016}. For brevity, we set $ \mathcal{M}_{0}=\mathcal{P}_{\mathbb{P}}^{\infty}$. For $\mathbb{Q}\in\mathcal{P}^\infty_{\mathbb{P}}$, we use in the following the sets $D_1(\mathbb{Q})$ and $D_2(\mathbb{Q})$ given by
        $$D_1(\mathbb{Q}) = \left\{p\in \left[0,\frac{1}{2}\right]\,\middle|\,\frac{\esssup\left(\frac{\mathrm{d}\mathbb{Q}}{\mathrm{d}\mathbb{P}}\right)}{\essinf\left(\frac{\mathrm{d}\mathbb{Q}}{\mathrm{d}\mathbb{P}}\right)}\leq \frac{1-p}{p}\right\}=[0,c(\mathbb{Q})]$$
        and
        $$D_2(\mathbb{Q}) = \left\{p\in \left[\frac{1}{2},1\right]\,\middle|\,\frac{\esssup\left(\frac{\mathrm{d}\mathbb{Q}}{\mathrm{d}\mathbb{P}}\right)}{\essinf\left(\frac{\mathrm{d}\mathbb{Q}}{\mathrm{d}\mathbb{P}}\right)}\leq \frac{p}{1-p}\right\}=[1-c(\mathbb{Q}),1].$$
        Note that $c(\mathbb{Q})\leq \frac{1}{2}$ for every $\mathbb{Q}\in\mathcal{P}^\infty_{\mathbb{P}}$. Finally, we obtain that
        \begin{align*}
            \aerm(X) 
            &= \max\left\{\sup_{p\in \left[0,\frac{1}{2}\right]}\left\{ \inf_{\mathbb{Q}\in \mathcal{M}_p}\{E_{\mathbb{Q}}[X]\}-g(p)\right\},\sup_{p\in \left[\frac{1}{2},1\right)}\left\{ \max_{\mathbb{Q}\in \mathcal{M}_p}\{E_{\mathbb{Q}}[X]\}-g(p)\right\}\right\}\\
            &=\max\left\{\sup_{p\in \left[0,\frac{1}{2}\right]}\left\{ \inf_{\substack{\mathbb{Q}\in\mathcal{P}^\infty_{\mathbb{P}}\\ p\in D_1(\mathbb{Q})}}\{E_{\mathbb{Q}}[X]-g(p)\}\right\},\sup_{p\in \left[\frac{1}{2},1\right)}\left\{ \sup_{\substack{\mathbb{Q}\in\mathcal{P}^\infty_{\mathbb{P}}\\ p\in D_2(\mathbb{Q})}}\{E_{\mathbb{Q}}[X]-g(p)\}\right\}\right\}\\
            &=\max\left\{\sup_{p\in \left[0,\frac{1}{2}\right]}\left\{ \inf_{\substack{\mathbb{Q}\in \mathcal{P}^\infty_{\mathbb{P}}\\ p\in[0,c(\mathbb{Q})]}}\{E_{\mathbb{Q}}[X]-g(p)\}\right\},\sup_{\mathbb{Q}\in\mathcal{P}^\infty_{\mathbb{P}}}\left\{E_{\mathbb{Q}}[X]-\inf_{p\in D_2(\mathbb{Q})}g(p)\right\}\right\}\\
            &= \max\left\{\sup_{p\in \left[0,\frac{1}{2}\right]}\left\{ \inf_{\mathbb{Q}\in \mathcal{P}^\infty_{\mathbb{P}}}\{E_{\mathbb{Q}}[X]-g(p)+\delta_{[0,c(\mathbb{Q})]}(p)\}\right\},\sup_{\mathbb{Q}\in\mathcal{P}^\infty_{\mathbb{P}}}\{E_{\mathbb{Q}}[X]-g(1-c(\mathbb{Q}))\}\right\}.
        \end{align*}
    \end{proof}

    \begin{remark}
        The previous proof shows that the infimum in the dual representation can be replaced by a minimum if $p\in\left(0,\frac{1}{2}\right]$.
    \end{remark}

        Typical levels for VaR applied in practice are above
        $90\%$. Empirically findings show that for expectiles we need levels larger than $90\%$, if the expectile should be equal to the VaR for some fixed level above $90\%$, see ~\cite[Figure 2.2]{Phd}. Hence, to avoid that levels below $90\%$ influence the outcome of the AERM we can e.g.~require that $g(p)=0$ for $p\leq \frac{1}{2}$. Then, the dual representation in Proposition~\ref{prop:dual_expectiles} reduces to      
        $$\sup_{\mathbb{Q}\in\mathcal{P}^\infty_{\mathbb{P}}}\{E_{\mathbb{Q}}[X]-g(1-c(\mathbb{Q}))\}.$$ 
       
       Here, the penalty term in the dual representation is  $g(1-c(\mathbb{Q}))$. This penalty term is minimized iff $c(\mathbb{Q})\in (0,1]$ is maximized. The latter becomes maximal for $\mathbb{Q}=\mathbb{P}$. For  $\mathbb{Q}\neq \mathbb{P}$, we have $c(\mathbb{Q})\leq c(\mathbb{P})$ and hence, the penalty term increases. Note that a decrease in $c$ occurs if either the essential infimum of $\frac{\mathrm{d}\mathbb{Q}}{\mathrm{d}\mathbb{P}}$ decreases or the essential supremum of $\frac{\mathrm{d}\mathbb{Q}}{\mathrm{d}\mathbb{P}}$ increases. Such  changes measure the difference between the potential model $\mathbb{Q}$ and the real-world measure $\mathbb{P}$. 
       
       In addition, the proof of Proposition~\ref{prop:dual_expectiles} shows that the AERM-dual representation has an analogous form as the dual representation of expectiles, but the target risk profile leads to an additional penalty term. Even more, the second term in the maximum in the dual representation of the AERM is of the same form as the dual representation of the adjusted ES, see~\cite[Proposition 3.7]{AES}. This is a consequence of the fact that the ES is coherent and that expectiles are coherent for levels larger than $\frac{1}{2}$. As the proof shows, the coherence of the expectiles gives us that the supremum over all levels between $\frac{1}{2}$ and $1$ is attained in the point $1-c(\mathbb{Q})$, given $\mathbb{Q}\in\mathcal{P}^\infty_{\mathbb{P}}$. We illustrate this in the next example.

    \begin{example}\label{ex dual}
        Let $g\in\mathcal{G}_0$ such that $g(p)=0$ for $p\in[0,0.95]$ and $g(p)=\infty$ otherwise. Then, we know for $X\in L^1$ that $\aerm(X)=e_{0.95}(X)$. Now, by the assumption of an atomless probability space, we can choose $X\in L^1$ such that $\mathbb{P}(X=1)=\mathbb{P}(X=0)=\frac{1}{2}$. Then, the supremum of the dual representation is attained for $\mathbb{Q}\in\mathcal{P}^\infty_{\mathbb{P}}$ with $\frac{\mathrm{d}\mathbb{Q}}{\mathrm{d}\mathbb{P}} = 1.9\cdot\mathds{1}_{\{X=1\}}+0.1\cdot\mathds{1}_{\{X=0\}}$. Indeed, in this case, it holds that $1-c(\mathbb{Q})=0.95$. Now, on the one side, if we choose another measure $\bar{\mathbb{Q}}$ with $$\esssup\left(\frac{\mathrm{d}\bar{\mathbb{Q}}}{\mathrm{d}\mathbb{P}}\right)- \essinf\left(\frac{\mathrm{d}\bar{\mathbb{Q}}}{\mathrm{d}\mathbb{P}}\right) >\esssup\left(\frac{\mathrm{d}\mathbb{Q}}{\mathrm{d}\mathbb{P}}\right)- \essinf\left(\frac{\mathrm{d}\mathbb{Q}}{\mathrm{d}\mathbb{P}}\right),$$ then this case is irrelevant, because $g\left(1-c\left(\bar{\mathbb{Q}}\right)\right)=\infty$. On the other side, if we choose a measure  $\bar{\mathbb{Q}}$ such that the inverse inequality holds, we do not maximize the expected value.
    \end{example}

    For the first term in the maximum of the dual representation of the AERM, it is not possible to interchange the supremum and the minimum, as we show in the next example.

    \begin{example}
        We show that there exists $X\in L^1$ such that
        $$\sup_{p\in \left[0,\frac{1}{2}\right]}\left\{ \inf_{\mathbb{Q}\in \mathcal{P}^\infty_{\mathbb{P}}}\{E_{\mathbb{Q}}[X]-g(p)+\delta_{[0,c(\mathbb{Q})]}(p)\}\right\}<\inf_{\mathbb{Q}\in \mathcal{P}^\infty_{\mathbb{P}}}\left\{\sup_{p\in \left[0,\frac{1}{2}\right]} \{E_{\mathbb{Q}}[X]-g(p)+\delta_{[0,c(\mathbb{Q})]}(p)\}\right\}.$$
        To do so, we choose $g\in\mathcal{G}_0$ with $g(p)=1$ for all $p\in\left(0,\frac{1}{2}\right)$ and $g(p)=\infty$ otherwise. As in the previous example, by the assumption of an atomless probability space, we can choose $X\in L^1$ with $\mathbb{P}(X=1)=\mathbb{P}(X=-1)=\frac{1}{2}$. 
        Then, for every $\mathbb{Q}\in \mathcal{P}^\infty_{\mathbb{P}}$ with $c(\mathbb{Q})<\frac{1}{2}$ we obtain that 
        $$\sup_{p\in \left[0,\frac{1}{2}\right]} \{E_{\mathbb{Q}}[X]-g(p)+\delta_{[0,c(\mathbb{Q})]}(p)\}=\infty.$$
        If $c(\mathbb{Q})=\frac{1}{2}$, then it has to hold that $\esssup\left(\frac{\mathrm{d}\mathbb{Q}}{\mathrm{d}\mathbb{P}}\right)=\essinf\left(\frac{\mathrm{d}\mathbb{Q}}{\mathrm{d}\mathbb{P}}\right)$ and so, by $ E\left[\frac{\mathrm{d}\mathbb{Q}}{\mathrm{d}\mathbb{P}}\right]=1$ we have that $\mathbb{Q}=\mathbb{P}$ . Hence, it follows that 
        $$\inf_{\mathbb{Q}\in \mathcal{P}^\infty_{\mathbb{P}}}\left\{\sup_{p\in \left[0,\frac{1}{2}\right]} \{E_{\mathbb{Q}}[X]-g(p)+\delta_{[0,c(\mathbb{Q})]}(p)\}\right\} = E_{\mathbb{P}}[X]=0.$$

        By the ordering of the family of expectiles, see~\cite[Proposition 5]{BELLINI201441}, we obtain that
         \begin{align*}
            \sup_{p\in \left[0,\frac{1}{2}\right]}\left\{ \inf_{\mathbb{Q}\in \mathcal{P}^\infty_{\mathbb{P}}}\{E_{\mathbb{Q}}[X]-g(p)+\delta_{[0,c(\mathbb{Q})]}(p)\}\right\}            &=\sup_{p\in \left[0,\frac{1}{2}\right]}\{ e_p(X)-g(p)\}\\
            &=\max\{e_0(X),e_\frac{1}{2}(X)-1\}\\
            &=-1\\
            &<0=\inf_{\mathbb{Q}\in \mathcal{P}^\infty_{\mathbb{P}}}\left\{\sup_{p\in \left[0,\frac{1}{2}\right]} \{E_{\mathbb{Q}}[X]-g(p)+\delta_{[0,c(\mathbb{Q})]}(p)\}\right\}.
        \end{align*}
    \end{example}

    \begin{remark}
    In the case of the SCRM and CRM we cannot apply convex duality arguments directly, due to the missing convexity of $\VaR$ and $\RVaR$. But, both risk measures are star-shaped. Therefore, we can apply the dual representation for star-shaped monetary risk measures, given in~\cite[Proposition 8]{laeven2023dynamicreturnstarshapedrisk}. To do so, assume an arbitrary extended real-valued function $f$ with domain $\mathcal{Y}$. The effective domain of $f$ is $\dom(f)=\{Y\in\mathcal{Y}\,|\,f(Y)<\infty\}$. Then, for every $Z\in \dom(f)$ we define the function $f^Z$ by $f^{Z}(Y)=\alpha f(Y)$, if there exists an $\alpha\in[0,1]$ with $Y=\alpha Z$ and $f^Z(Y) = \infty$, otherwise. Then, for $\mathcal{X}=L^1$, a family of risk measures $\Rho=\{\rho_p\}_{p\in[0,1]}$ and $g\in\mathcal{G}_0$ such that $\adjustedRiskMeasure$ is star-shaped, we obtain by~\cite[Proposition 8]{laeven2023dynamicreturnstarshapedrisk} that for every $X\in L^1$ it holds that
    \begin{align*}
        \adjustedRiskMeasure(X) 
        = \min_{Z\in L^1}\left\{\adjustedRiskMeasure^{Z}(X)\right\}
        =\min_{Z\in L^1}\left\{\sup_{Y\in L^{\infty}}\left\{E(YX)-\max\{0,E[ZY]-\adjustedRiskMeasure(Z)\}\right\}\right\}.
    \end{align*}
    \end{remark}
    
    \section{Case study}\label{case}
    
    In this section, we test the performance of the new risk measures from Section 3 for the S$\&$P 500 index from January $03$, $2000$ to February $08$, $2024$. In addition, we consider individual stocks out of the S$\&$P 500\footnote{This data was obtained from~\url{https://de.finance.yahoo.com/}.}.
    
    The concrete choice of the target risk profile $g$ is crucial for an adjusted risk measure. We test the following concrete possibilities: In Section~\ref{sec:caseStudy_stepFunctions}, we model $g$ as a step function. In contrast, $g$ can also be obtained as risk measures of a benchmark random variable $Z\in L^1$. For instance, in case of the adjusted ES one can use $g:p\mapsto \ES_{p}(Z)$. In such cases, we say $g$ is a benchmark profile. We use this methodology in Sections~\ref{sub bench} and~\ref{sub rol}. Both variants, step functions and benchmark profiles, are suggested in~\cite{AES}, but they only illustrate these variants in a short real-world example, see~\cite[Figure 2]{AES}. Therefore, we aim for a more comprehensive case study.
    
    We reevaluate the risk measures over time by calculating them for a rolling time window of $60$ days, i.e.,~we compute the empirical estimators of VaR, ES, RVaR and expectiles for a rolling window of $60$ data points and calculate the values of the adjusted risk measures with the help of a numerical procedure. For simplicity, let us introduce the empirical estimators for the first time window, for which we denote the data points by  $x_1,\dots,x_{60}$. As data points, we use the negative daily log-returns of the underlying time series. In other words, the random variable $X$ describes the distribution of the negative daily log-return. This is in line with~\cite[Figure 2]{AES}.
    
    We use the same empirical estimators for VaR, ES and RVaR as in~\cite[Examples 2.4, 2.5, 3.2.3]{Robust}. To introduce them, we denote by $x_{(i)}$ the $i$-th smallest element of the data set $\{x_1, \dots , x_{60}\}$. Then, 	for a level $p\in(0,1)$, we obtain
	\begin{align*}
	    \widehat{\text{VaR}}_p(X) &= x_{(\lfloor 60p\rfloor +1)},\\
            \widehat{\text{ES}}_p(X) &= \frac{1}{60-60p}\left(\left(\sum_{i=\lfloor60p\rfloor+2}^{60}x_{(i)}\right)+x_{(\lfloor 60p\rfloor +1)}(\lfloor 60p\rfloor+1-60p )\right).
	\end{align*}

    For $\alpha_1,\alpha_2\in(0,1)$ with $\alpha_1<\alpha_2$ we use a discrete version of the risk measure, where the interval $[\alpha_1,\alpha_2]$ is divided in $20$ equidistant intervals via the points $$u_0 = \alpha_1<u_1 = \alpha_1+\frac{1}{19}(\alpha_2-\alpha_1)<\dots<u_{18} =\alpha_1+\frac{18}{19}(\alpha_2-\alpha_1) <u_{19} = \alpha_2.$$ 
    
    Then, the empirical estimator of the RVaR is
	\begin{align*}
	    \widehat{\text{RVaR}}_{\alpha_1,\alpha_2}(X) = \frac{1}{20}\sum_{i=0}^{19}\widehat{\text{VaR}}_{u_i}(X).
	\end{align*}

    The empirical estimator of the expectile $e_q(X)$ with $q\in(0,1)$ is calculated as the root of the function
	$$f(q) = qE[\max\{X-e_q(X),0\}]-(1-q)E[-\min\{X-e_q(X),0\}],$$
    where the expectations are estimated from the  empirical CDF obtained from the data set $\{x_1, \dots , x_{60}\}$. 
    
    \subsection{Step functions}\label{sec:caseStudy_stepFunctions}
    
    In \cite{AES}, the authors illustrate the adjusted ES for the following target risk profile:
    \begin{align}\label{sec:step_function_case_study}
        g(p)=0.01\cdot\mathds{1}_{(0.95,0.99]}(p)+\infty\cdot\mathds{1}_{(0.99,1]}(p).
    \end{align}

    Then, for $X \in L^1$ it holds that
	$$\text{ES}^g(X) = \max\{\text{ES}_{0.95}(X),\text{ES}_{0.99}(X)-0.01\}.$$
	
    We compare this with the SCRM, the FCRM and the AERM, in which we choose the probability levels such that we obtain the following representations for $X\in L^1$ (cf.~Example~\ref{example}):
	\begin{enumerate}[(i)]
	    \item $\scrm(X) = \max\{\text{VaR}_{0.95}(X),\text{ES}_{0.99}(X)-0.01\},$
            \item $\fcrm(X) = \max\{\text{RVaR}_{0.95,0.99}(X),\text{ES}_{0.99}(X)-0.01\},$
            \item $\aerm(X) = \max\{e_{0.95}(X),e_{0.99}(X)-0.01\}.$
	\end{enumerate}

    To compare one of these new adjusted risk measures $\adjustedRiskMeasure$ with the known adjusted ES, we also compare the relative difference between them, i.e.,~we calculate the value
	$$\frac{\adjustedRiskMeasure(X)-\text{ES}^g(X)}{\text{ES}^g(X)}.$$

    \begin{figure}[h]
		\begin{center}
			\includegraphics[scale=0.7]{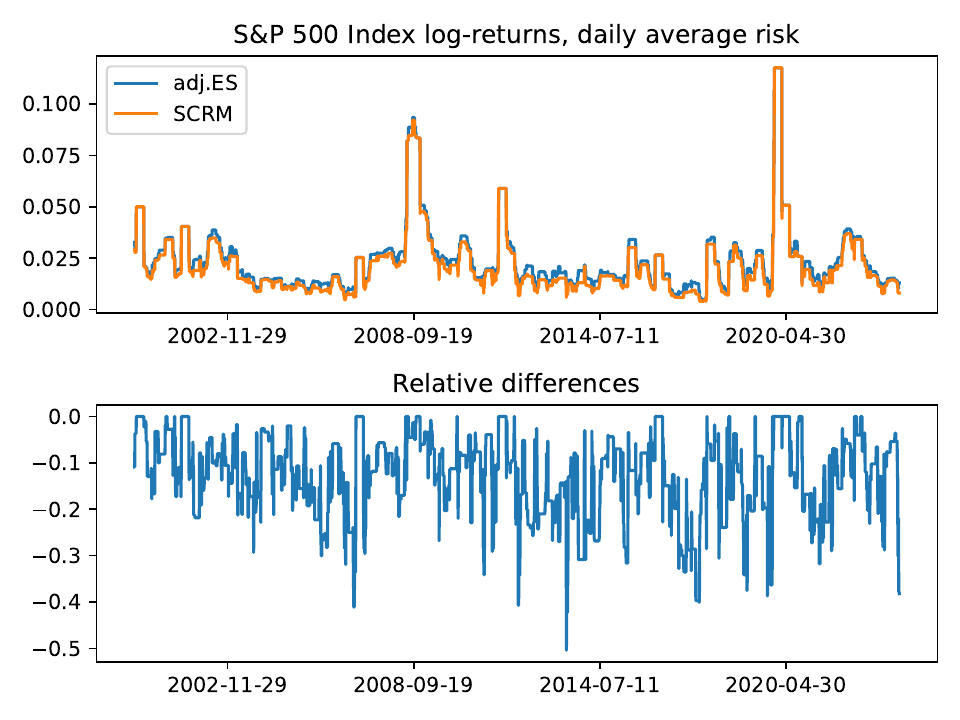}
		\end{center}
		\caption{\footnotesize{\textit{First plot:} SCRM and adjusted ES (adj.~ES). \textit{Second plot:} Relative difference between SCRM and adjusted ES.}}
		\label{1}
	\end{figure}

    For brevity, in this subsection we only compare the SCRM and the adjusted ES in Figure~\ref{1}. The SCRM is always smaller or equal than the adjusted ES, because both risk measures can either coincide in the value $\text{ES}_{0.99}(X)-0.01$ or attain $\text{ES}_{0.95}(X)$, in case of the adjusted ES, and $\text{VaR}_{0.95}(X)$, in case of the SCRM. In the latter case, it holds that $\text{VaR}_{0.95}(X)\leq \text{ES}_{0.95}(X)$. Furthermore, the relative difference between SCRM and adjusted ES is on average $-12.6\%$ and the corresponding median is $-11.3\%$. The huge relative differences between $-60\%$ and $-40\%$ arise at times in which the risk measures are small. In these cases, we have two possibilities for the values of the  adjusted ES and the SCRM. First, the SCRM is given by $\text{VaR}_{0.95}(X)$ and the adjusted ES is given by $\text{ES}_{0.95}(X)$. Second the SCRM is given by $\text{ES}_{0.99}(X)-0.01$ and the adjusted ES is given by $\text{ES}_{0.95}(X)$. In the first case, the tail blindness of the VaR is the reason why the SCRM reacts slower than the adjusted ES. To be precise, if the time interval for the calculation of the SCRM is shifted, even if a new data point is large, then the robustness of the VaR estimator leads only to small changes in the SCRM. In contrast, for the second case, when the maximum for the SCRM is attained in $\text{ES}_{0.99}(X)-0.01$, then the missing robustness of the ES estimator makes the SCRM more sensitive to a shift of the underlying time interval.
    
    Now we turn to the analysis of times in which the risk measures attain large values. To do so, on the left-hand side in Figure~\ref{2} we illustrate the financial crisis in $2008$. Furthermore, on the right-hand side in Figure~\ref{2} we test a different jump size of $0.03$ for the target risk profile than the one used in the step function from Equation~(\ref{sec:step_function_case_study}). Comparing these two cases for the SCRM, Figure~\ref{2} allows us to conclude that most of the time during the crisis the SCRM is given by the ES at level $0.99$ minus the jump size, i.e.,~either $0.01$ or $0.03$. However, the largest SCRM value around September 18 is equal in both cases. This means that the SCRM is equal to $\text{VaR}_{0.95}(X)$. We conclude that in times of crisis there is no clear tendency for which term the maximum in the calculation of the SCRM is attained. 
    
    In addition, the difference between SCRM and adjusted ES is much more pronounced in the case of the larger jump size. This means that the adjusted ES leads to a more conservative risk assessment than the SCRM.
	
    \begin{figure}[h]
		\begin{center}
			\includegraphics[scale=0.77]{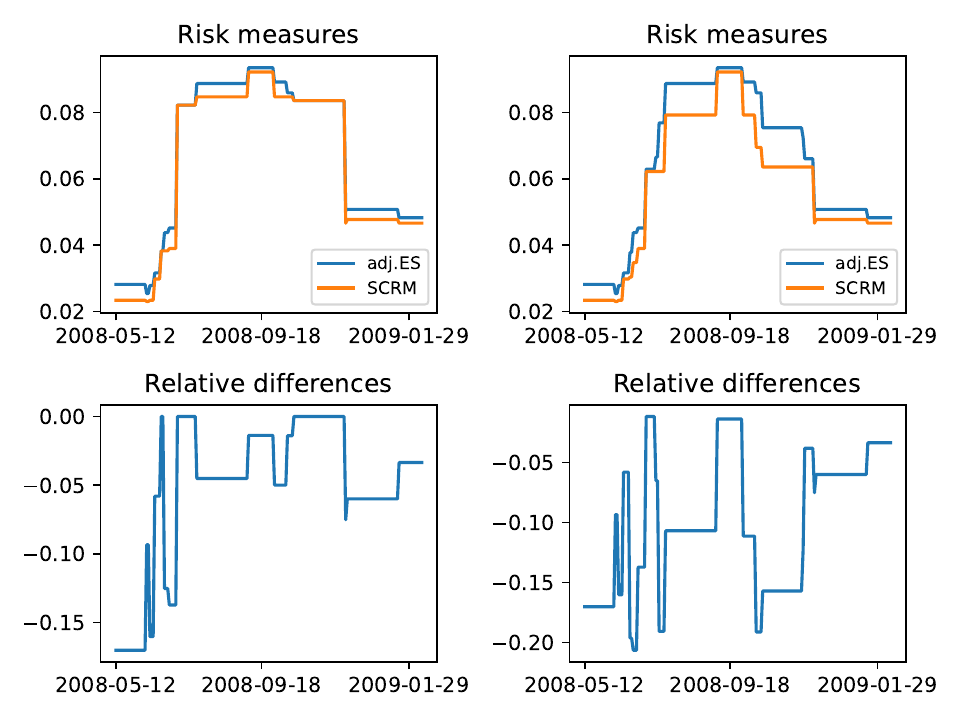}
		\end{center}
		\caption{\footnotesize{SCRM and adjusted ES (adj.~ES) in times of the financial crisis 2008 (upper row). On the left-hand side, we use the step function from Equation~(\ref{sec:step_function_case_study}). On the right-hand side we changed the jump size of this step function to $0.03$. The bottom line shows the  relative differences between SCRM and adjusted ES. }}
		\label{2}
	\end{figure}

    \subsection{Benchmark profiles}\label{sub bench}
	Now we calculate adjusted risk measures for two different stocks out of the S$\&$P $500$. These stocks are from different market segments, namely Microsoft (industry: technology) and Boeing (industry: aerospace). Key statistics of these stocks are given in Table~\ref{table}.

    \begin{table}[H]
    $$\begin{tabular}{|c||c|c|c|c|c|c|c|c|c|}
        \hline
        Asset& mean & median & std.~dev.& 25\%-qnt.& 75\%-qnt. & corr. M& corr. B& corr. SPX\\
        \hline
        \hline
        M & -0.040\% & -0.036\% & 1.9\% & -0.94\% & 0.82\% & 1 & 0.38&0.71 \\
        \hline
        B & -0.034\% & -0.061\% & 2.2\% & -1.10\% & 0.98\% & 0.38 & 1&0.61 \\
        \hline
        SPX & -0.020\% & -0.059\% & 1.2\% & -0.59\% & 0.49\% & 0.71 &0.61& 1 \\
        \hline
    \end{tabular}$$
    \caption{\footnotesize{Statistical quantities of negative daily log-returns of Microsoft (M), Boeing (B) and the S$\&$P $500$ index (SPX) between January $03$, $2000$ and February $08$, $2024$. Here, ``std.~dev.'', ''qnt.''~and ``corr.''~stands for standard deviation, quantile and correlation.}}
    \label{table}
\end{table}

    The target risk profiles are given as benchmark profiles in which we use the S$\&$P $500$ as benchmark random variable. For instance, in case of the AERM, $g$ at level $p$ is given by the expectile of the S$\&$P $500$ at level $p$. The benchmark  profile of the SCRM is based on the VaR up to a level of $0.99$. Further, instead of the FCRM, we use now the CRM, because the FCRM would reduce to a finite number of RVaRs, compare also the Definitions~\ref{CRM} and~\ref{FCRM}. In this case, VaRs are replaced by RVaRs with the following levels: $p_0=0.0001,\,p_1=0.165, \,p_2=0.33,\, p_3=0.495,\, p_4=0.66,\, p_5=0.825,\,p_6=0.99.$
    
    We already saw that the SCRM and the adjusted ES attain large values in times of crisis in which the market is usually more volatile. Motivated by this fact, we use different time periods to obtain empirical distributions representing different benchmark random variables. These time periods are chosen out of the first ten years of our underlying data:
    \begin{enumerate}[(1)]
        \item \textit{Low volatility frame:} December 26, 2003 to December 21, 2006,
        \item \textit{Medium volatility frame:} January 03, 2000 to January 03, 2002,
        \item \textit{High volatility frame:} May 15, 2007 to May 11, 2009.
    \end{enumerate}
     
    To ensure that the benchmark profile $g$ is an element of the set $\mathcal{G}_0$, we stipulate $g(p) =0$ if the estimator of the monetary risk measure is less than zero, i.e.,~$\widehat{\rho_p}(X)<0$. For brevity, we only present the benchmark profiles of the adjusted ES and the AERM for the low and the high volatility frames (Figure~\ref{Lowvolag}), which can be found in Appendix~\ref{sec:figures}. The estimators for the adjusted risk measures are then obtained by discretizing the interval of all levels $[0,1]$ by steps of size $0.02$ and calculating the maximum with respect to this grid, i.e.~the maximum of the underlying monetary risk measure minus the target risk profile. Regarding the discussion in Section~\ref{sub}, i.e.,~to guarantee finiteness and continuity of the adjusted ES and AERM, the calculation is restricted to levels between $0.01\%$ and $99.99\%$. For the composed risk measures, we set the underlying family of risk measures to be equal to the ES for levels above $r=p_n=99\%$, corresponding to Definition~\ref{def SCRM} and Definition~\ref{CRM}.
        
    Now, we are in the position to calculate the adjusted risk measures. We evaluate them in the time period from January 04, 2010 to February 08, 2024, see Figure~\ref{10}. The maximum values in Figure~\ref{10} are larger in case of Boeing. The maximum values are observed in the Corona crisis in January 2020. Here, we see that the SCRM for Boeing is twice the SCRM for Microsoft. This can be explained as follows: The larger standard deviation of Boeing in Table \ref{table} means that the log-returns spread more widely than in case of Microsoft, which explains the larger peaks in case of Boeing. 

    \begin{figure}[!]
		\begin{center}
			\includegraphics[scale=0.83]{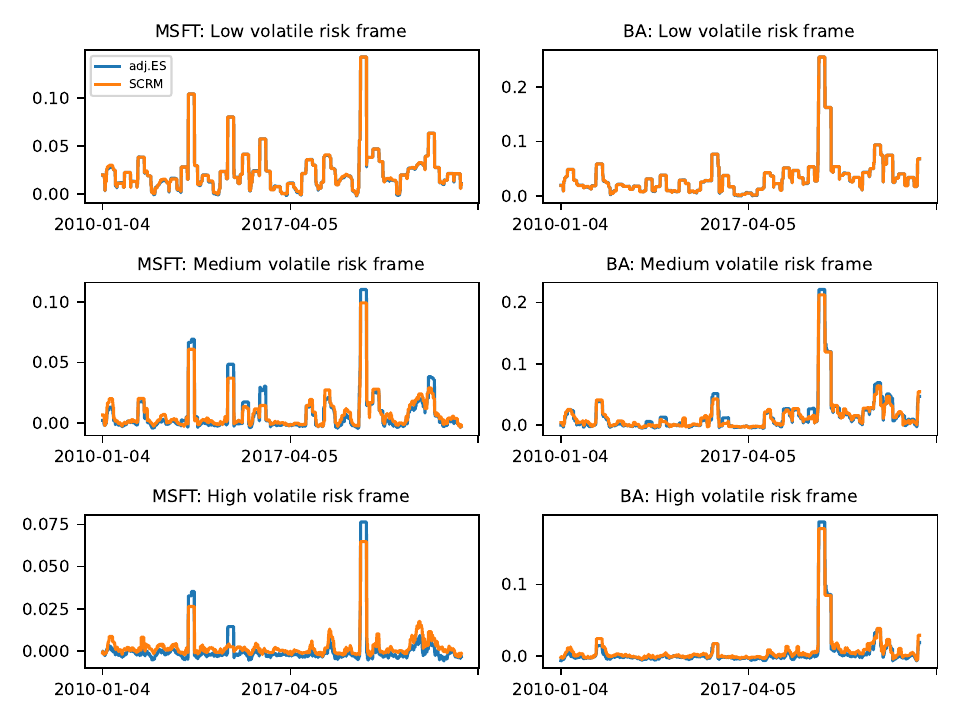}
			\includegraphics[scale=0.83]{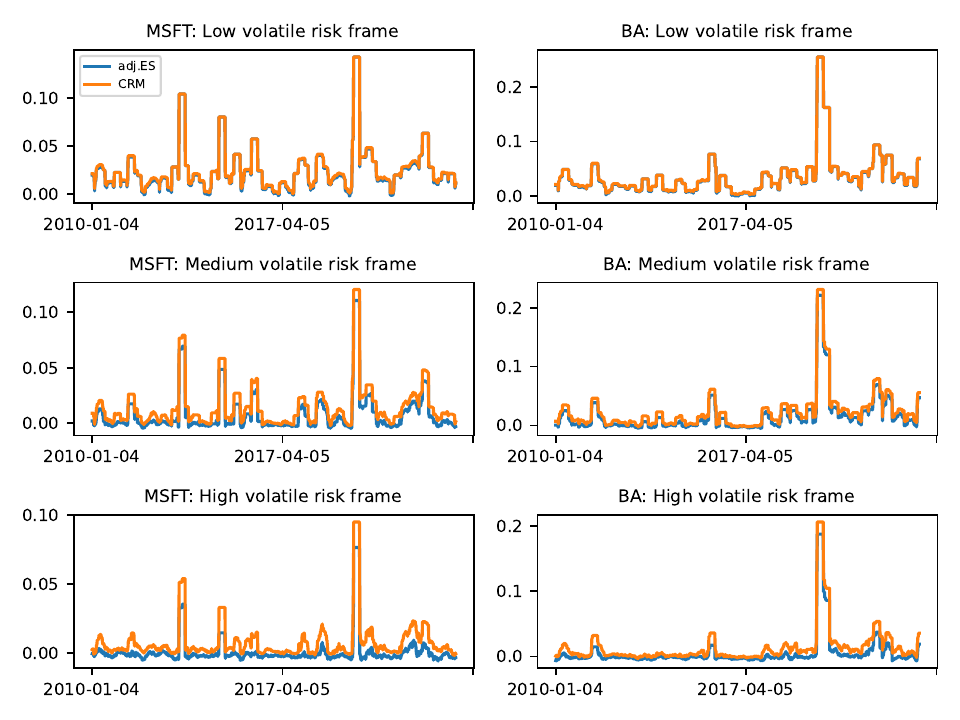}
		\end{center}
		\caption{\footnotesize{SCRM, CRM and adjusted ES (adj.~ES) for Microsoft (MSFT) and Boeing (BA) and for three different benchmark profiles based on the S$\&$P $500$.}}
		\label{10}
    \end{figure}
    
    Furthermore, for both stocks, the risk measures are decreasing from the low volatility frame towards the high volatility frame. The reason is that the choice of the volatility time frame significantly influences the magnitude of the benchmark profile. From Figure~\ref{Lowvolag} we see that a more volatile time frame leads to larger values of the benchmark profile. For instance, the benchmark profile calibrated for the low volatility frame affects the adjusted risk measure in the sense that it reacts more sensitive to small changes in the stock price movements. This finding gives rise to different use cases for adjusted risk measures depending on the benchmark profile: Small benchmark profiles (low volatility frame) are useful for a conservative calculation of required key figures to measure the risk. In contrast, for large benchmark profiles (high volatility frame), the adjusted risk measures can be used to detect times of crises in the underlying data, because they only detect peaks and ignore usual market fluctuations.   
    
    In addition, focusing on the low volatility frame, the SCRM and the CRM are most of the time equal to the adjusted ES. Hence, the usage of the SCRM instead of the adjusted ES is only significant if the target risk profiles are based on a time period in which the market is more volatile.

    If we look at the high volatility frame, then another observation is that the SCRM and the CRM behave differently. To be precise, the SCRM is larger than the adjusted ES, for small values, and lower than the adjusted ES, for large values. In contrast, the CRM is always larger than the adjusted ES. The larger value of the CRM means that the difference between the RVaR of a stock and the RVaR of the index is larger than the difference of the corresponding VaR values. This leads us to the conclusion that the underlying family of risk measures should be chosen according to the personal risk aversion. For instance, more risk averse investors should use the CRM to obtain larger values for the corresponding adjusted risk measure.  
    
    For completeness, we mention that similar conclusions can be drawn for the AERM, see Figure~\ref{Stockbenchmarkaerm1} in Appendix~\ref{sec:figures}.
    
    \subsection{Benchmark profiles reevaluated over time}\label{sub rol}

    So far, we always prefix the target risk profile, i.e.,~it does not change over time. This is for instance the case if a regulator has to impose a target risk profile for legal regulation. From now on, we test the impact of daily updating the target risk profile due to the latest observations. Indeed, it is common practice to recalibrate model parameters for portfolio and risk management on a daily basis.
    
	For our purpose, it is necessary to fix the length of the time interval that is used to obtain the target risk profile. We use two concrete choices, namely the past $60$ and the past $200$ days.
    As before, we calculate the adjusted risk measures for the two stocks. The benchmark profiles are calibrated with respect to the S$\&$P 500. 

	Figure~\ref{11} shows the SCRM, the CRM and the AERM compared to the adjusted ES for the Microsoft stock. On the left-hand side in Figure~\ref{11}, the rolling time window of $60$ days and on the right-hand side, the  rolling time window of $200$ days is used. For comparison, we also plot the adjusted risk measures for prefixed benchmark profiles based on a low volatility frame, recall the results from Section~\ref{sub bench}.

    \begin{figure}[h]
		\begin{center}
			\includegraphics[scale=0.85]{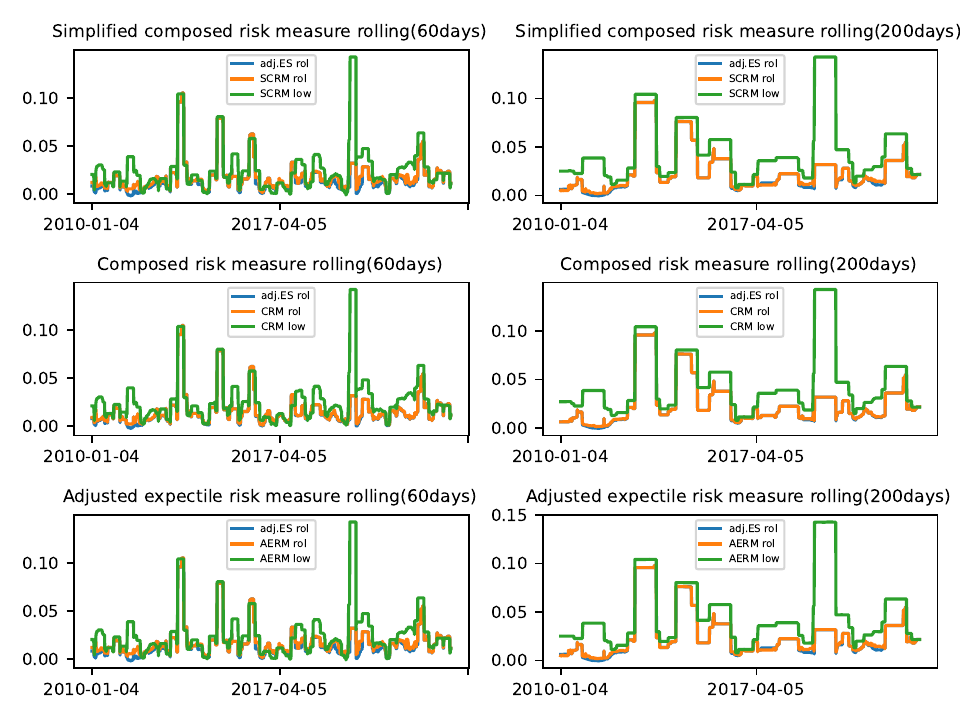}
		\end{center}
        \caption{\footnotesize{Reevaluated adjusted risk measures (``rol'') for Microsoft plotted against the benchmark adjusted risk measure based on the low volatility  frame (``low'') from Section~\ref{sub bench}.}}
		\label{11}
	\end{figure}    

     The adjusted ES is at many times equal to the SCRM, the CRM and the AERM, i.e.,~the graphs are overlapping. 
     Furthermore, the adjusted risk measures with reevaluated benchmark profiles admit smaller values than the adjusted risk measures for the low volatility frame. The aforementioned effect is amplified during the peak of the Corona crisis. For instance, in the case of $200$ days, the adjusted risk measures based on the low volatility frame are roughly four to five times larger than the reevaluated counterparts with a peak above~$0.14$. This means that the reevaluated adjusted risk measures are less sensitive towards strong fluctuations of the underlying time series data. In addition, this effect is slightly mitigated in case of using $60$ days instead of $200$ days: during the Corona crisis, the adjusted risk measures for the $60$ days time window are around $0.04$, which is larger than $0.03$ in case of the $200$ days time window. Hence, the risk measures based on $200$ days are slightly less sensitive to detect a crisis in the underlying time series data.

    To understand the different magnitudes during the Corona crisis, we look at the level that solves the underlying maximization problem of an adjusted risk measure in Figure~\ref{coronap}. We see that the optimal level of the SCRM for the low volatility frame is at many times  larger than the ones for the reevaluated versions. This stems from the fact that in case of the low volatility frame, the benchmark profile admits small values. So, the difference between risk measure and benchmark profile is especially large for levels close to $1$. Note, for a level equal to one, the ES is the largest value (i.e.~the essential supremum) of the last $60$ data points.

    \begin{figure}[h]
		\begin{center}
			\includegraphics[scale=0.75]{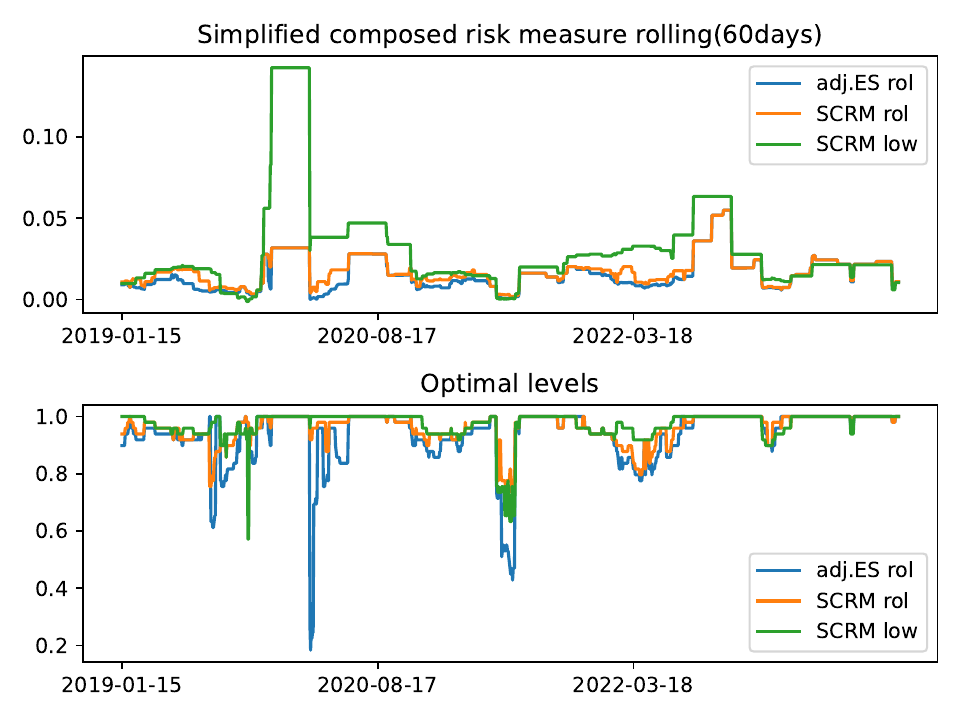}
		\end{center}
        \caption{\footnotesize{{\textit{First plot:} Reevaluated adjusted risk measures based on the last $60$ days (``rol'') for Microsoft plotted against the adjusted risk measure for the low volatility  frame (``low'') from Section~\ref{sub bench}. \textit{Second plot:} Optimal levels of the adjusted risk measures.}}}
		\label{coronap}
	\end{figure} 
    
    As a final part, we test if adjusted risk measures are useful to detect outliers between different financial markets. To do so, we consider daily log-returns of two different stock indices, namely the Eurostoxx $50$ (for the European market)\footnote{This data was also obtained from~\url{https://de.finance.yahoo.com/}} and the S$\&$P $500$. We use the Eurostoxx $50$ as input for the adjusted risk measure with reevaluated benchmark profile based on the S$\&$P $500$. Recall that the adjusted risk measures are based on negative daily log-returns. We focus on data between January 15, 2019 and February 08, 2024.

    \begin{figure}[h]
		\begin{center}
			\includegraphics[scale=0.85]{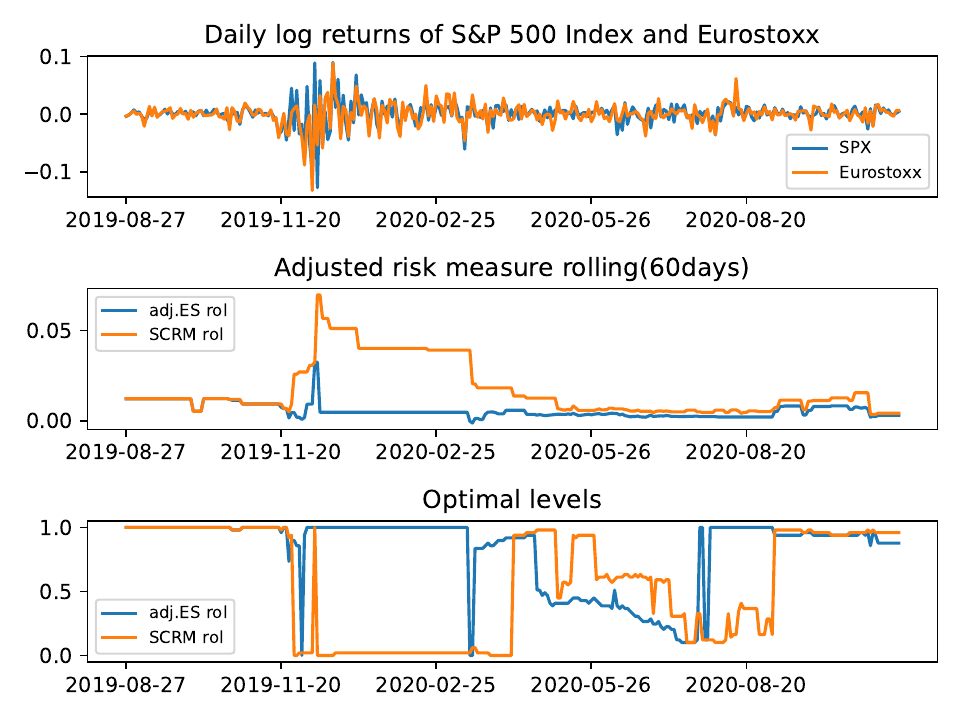}
		\end{center}
        \caption{\footnotesize{{\textit{First plot:} Daily log-returns of the S$\&$P $500$ and the Eurostoxx $50$. \textit{Second plot:} Adjusted risk measures of the Eurostoxx $50$ index with a reevaluated benchmark profile based on the past $60$ days of the S$\&$P $500$. \textit{Third plot:} Optimal levels of the adjusted risk measures of the Eurostoxx $50$ index with a reevaluated benchmark profile based on the past $60$ days of the S$\&$P $500$.}}}
		\label{Goldp}
	\end{figure}
    
    The adjusted ES and the SCRM are illustrated in the second plot in Figure~\ref{Goldp}. They admit a peak during the Corona crisis. This means that the two indices are maximally different at this time. The different magnitudes of the peaks for the SCRM and the adjusted ES rely on the fact that the SCRM detect either the difference of the largest daily log-returns, i.e.,~the difference in essential infima (for level $p=0$), or the difference of the smallest daily log-returns, i.e.,~the difference in essential suprema (for level $p=1$). In contrast, the adjusted ES only takes the largest daily log-returns via the expectation into account (for level $p=0$). Hence, the adjusted ES is mostly dominated by large negative outliers in the daily log-returns, which is in particular the case during the Corona crisis. For completeness, the optimal levels of the adjusted ES and the SCRM are illustrated in the third plot in Figure~\ref{Goldp}. Note, before the Corona crisis both risk measures have an optimal level of $p=1$, i.e.,~both risk measures detect an significant difference in the smallest daily log-returns between both indices. During the Corona crisis, the SCRM detects also large differences between the largest daily log-returns, i.e.,~the optimal level of the SCRM is given by $p=0$. So, the SCRM in combination with a benchmark profile reacts towards positive and negative outliers, while the adjusted ES only reacts against outliers given by small log-returns.

   Concluding, adjusted risk measures based on reevaluated benchmark profiles are significantly smaller than adjusted risk measures based on preset benchmark profiles in times of crisis. Further, the reevaluated versions are useful to detect differences between financial markets. Here, the SCRM detects positive and negative outliers. In contrast, the adjusted ES only focuses on negative outliers in the daily log-returns.

    \subsection{Summary of key findings}
    In Section~\ref{sub bench}, we saw that adjusted risk measures with benchmark profiles calibrated on a high volatility frame are useful to visualize times of crisis in the underlying data. In contrast, using a step function as target risk profile (Section~\ref{sec:caseStudy_stepFunctions}) or a benchmark profile calibrated on a low volatility frame (Section~\ref{sub bench}) is more sensitive to common market fluctuations. So, both options seem to be suitable for risk averse investors. In contrast, if we reevaluate the benchmark profiles over time (Section~\ref{sub rol}), then the adjusted risk measures underestimate times of crisis. So, they should only be considered by less risk averse investors. In addition, adjusted risk measures based on reevaluated benchmark profiles are useful to compare financial market indices. Here, the choice of the family of risk measures defines if positive and negative outliers are relevant (SCRM) or if the focus lies on negative outliers (adjusted ES).
    
    \section*{Data availability}

    How we obtained the time series data is explained in Section~\ref{case}.

\bibliographystyle{amsalpha}
\bibliography{main}

\newcommand{\etalchar}[1]{$^{#1}$}
\providecommand{\bysame}{\leavevmode\hbox to3em{\hrulefill}\thinspace}
\providecommand{\MR}{\relax\ifhmode\unskip\space\fi MR }
\providecommand{\MRhref}[2]{%
  \href{http://www.ams.org/mathscinet-getitem?mr=#1}{#2}
}
\providecommand{\href}[2]{#2}
\begin{thebibliography}{FKMM14}

\bibitem[AB06]{aliprantis_infinite_2006}
Charalambos~D. Aliprantis and Kim~C. Border, \emph{{Infinite {Dimensional}
  {Analysis}: {A} {Hitchhiker}'s {Guide}}}, 3 ed., Springer-Verlag, Berlin,
  2006.

\bibitem[ADEH99]{https://doi.org/10.1111/1467-9965.00068}
Philippe Artzner, Freddy Delbaen, Jean-Marc Eber, and David Heath,
  \emph{Coherent measures of risk}, Mathematical Finance \textbf{9} (1999),
  no.~3, 203--228.

\bibitem[AT02]{ACERBI20021487}
Carlo Acerbi and Dirk Tasche, \emph{{On the coherence of Expected Shortfall}},
  Journal of Banking and Finance \textbf{26} (2002), no.~7, 1487--1503.

\bibitem[BBM20]{LVaR}
Valeria Bignozzi, Matteo Burzoni, and Cosimo Munari, \emph{Risk measures based
  on benchmark loss distributions}, Journal of Risk and Insurance \textbf{87}
  (2020), no.~2, 437--475.

\bibitem[BCB12]{bcbs2012consultative}
BCBS, \emph{{Consultative Document, May 2012. Fundamental Review of the Trading
  Book: Basel committee on banking supervision}}, 2012.

\bibitem[BDB17]{expectiles}
Fabio Bellini and Elena Di~Bernardino, \emph{{Risk management with
  expectiles}}, The European Journal of Finance \textbf{23} (2017), no.~6,
  487--506.

\bibitem[BKMR14]{BELLINI201441}
Fabio Bellini, Bernhard Klar, Alfred Müller, and Emanuela {Rosazza Gianin},
  \emph{{Generalized quantiles as risk measures}}, Insurance: Mathematics and
  Economics \textbf{54} (2014), 41--48.

\bibitem[BMW22]{AES}
Matteo Burzoni, Cosimo Munari, and Ruodu Wang, \emph{{Adjusted Expected
  Shortfall}}, Journal of Banking and Finance \textbf{134} (2022), 106297.

\bibitem[CDS10]{Robust}
Rama Cont, Romain Deguest, and Giacomo Scandolo, \emph{{Robustness and
  sensitivity analysis of risk measurement procedures}}, Quantitative Finance
  \textbf{10} (2010), no.~6, 593--606.

\bibitem[EKT15]{VaRgegenES2}
Susanne Emmer, Marie Kratz, and Dirk Tasche, \emph{{What is the best risk
  measure in practice? A comparison of standard measures}}, Journal of Risk
  \textbf{18} (2015), 31--60.

\bibitem[ELW18]{RVaR2}
Paul Embrechts, Haiyan Liu, and Ruodu Wang, \emph{Quantile-based risk sharing},
  Operations Research \textbf{66} (2018), no.~4, 936--949.

\bibitem[EPR{\etalchar{+}}14]{VaRgegenES1}
Paul Embrechts, Giovanni Puccetti, Ludger Rueschendorf, Ruodu Wang, and
  Antonela Beleraj, \emph{An academic response to {Basel 3.5}}, Risks
  \textbf{2} (2014), 25--48.

\bibitem[ESW22]{doi:10.1287/opre.2021.2147}
Paul Embrechts, Alexander Schied, and Ruodu Wang, \emph{Robustness in the
  optimization of risk measures}, Operations Research \textbf{70} (2022),
  no.~1, 95--110.

\bibitem[FG21]{pelve}
Anna~Maria Fiori and Emanuela~Rosazza Gianin, \emph{Generalized {PELVE} and
  applications to risk measures}, European Actuarial Journal \textbf{13}
  (2021), 307--339.

\bibitem[FKMM14]{Farkas}
Walter Farkas, Pablo Koch-Medina, and Cosimo Munari, \emph{{Beyond
  cash-additive risk measures: when changing the numéraire fails}}, Finance
  and Stochastics \textbf{18} (2014), no.~1, 145–173.

\bibitem[FM02]{FREY20021317}
Rüdiger Frey and Alexander~J. McNeil, \emph{{VaR and {Expected Shortfall} in
  portfolios of dependent credit risks: conceptual and practical insights}},
  Journal of Banking and Finance \textbf{26} (2002), no.~7, 1317--1334.

\bibitem[FS16]{follmer_stochastic_2016}
Hans Föllmer and Alexander Schied, \emph{{Stochastic {Finance}: {An}
  {Introduction} in {Discrete} {Time}}}, De Gruyter, July 2016.

\bibitem[FZ21]{elictable}
Tobias Fissler and {Johanna F.} Ziegel, \emph{On the elicitability of {Range
  Value at Risk}}, Statistics and Risk Modeling \textbf{38} (2021), no.~1-2, 25
  -- 46.

\bibitem[Gul00]{guldimann2000story}
Till Guldimann, \emph{{The story of RiskMetrics}}, Risk \textbf{13} (2000),
  no.~1, 56--58.

\bibitem[HK24]{herdegen2024rhoarbitragerhoconsistentpricingstarshaped}
Martin Herdegen and Nazem Khan, \emph{$\rho$-arbitrage and $\rho$-consistent
  pricing for star-shaped risk measures}, Mathematics of Operations Research
  (2024), \url{https://doi.org/10.1287/moor.2023.0173}.

\bibitem[HRRS11]{hampel2011robust}
Frank~R. Hampel, Elvezio~M. Ronchetti, Peter~J. Rousseeuw, and Werner~A.
  Stahel, \emph{{Robust Statistics: The Approach Based on Influence
  Functions}}, Wiley Series in Probability and Statistics, Wiley, 2011.

\bibitem[KMM16]{KMM}
Pablo Koch-Medina and Cosimo Munari, \emph{Unexpected shortfalls of {Expected
  Shortfall}: Extreme default profiles and regulatory arbitrage}, Journal of
  Banking and Finance \textbf{62} (2016), 141--151.

\bibitem[KSZ14]{Kraetschmer2014}
Volker Krätschmer, Alexander Schied, and Henryk Zähle, \emph{{Comparative and
  qualitative robustness for law-invariant risk measures}}, Finance and
  Stochastics \textbf{18} (2014), no.~2, 271–295.

\bibitem[LGZ23]{laeven2023dynamicreturnstarshapedrisk}
Roger J.~A. Laeven, Emanuela~Rosazza Gianin, and Marco Zullino, \emph{Dynamic
  return and star-shaped risk measures via {BSDEs}}, arXiv,
  \url{https://doi.org/10.48550/arXiv.2307.03447}, 2023.

\bibitem[LW16]{ptail}
Fangda Liu and Ruodu Wang, \emph{{A theory for measures of tail risk}},
  Mathematics of Operations Research \textbf{46} (2016), no.~3, 1109--1126.

\bibitem[MW20]{Mao}
Tiantian Mao and Ruodu Wang, \emph{Risk aversion in regulatory capital
  principles}, SIAM Journal on Financial Mathematics \textbf{11} (2020), no.~1,
  169--200.

\bibitem[NP87]{1234}
Whitney~K. Newey and James~L. Powell, \emph{Asymmetric least squares estimation
  and testing}, Econometrica \textbf{55} (1987), no.~4, 819--847.

\bibitem[RM23]{RVaR}
Marcelo Righi and Fernanda Müller, \emph{{Range-based risk measures and their
  applications}}, ASTIN Bulletin \textbf{53} (2023), 1--22.

\bibitem[Rro14]{Phd}
Edit Rroji, \emph{Risk attribution and semi-heavy tailed distributions},
  Università degli Studi di Milano-Bicocca Dipartimento di Statistica e Metodi
  Quantitativi (2014).

\bibitem[RU02]{ROCKAFELLAR20021443}
R.Tyrrell Rockafellar and Stanislav Uryasev, \emph{{Conditional Value-at-Risk
  for general loss distributions}}, Journal of Banking and Finance \textbf{26}
  (2002), no.~7, 1443--1471.

\bibitem[Rü13]{RVaR3}
Ludger Rüschendorf, \emph{{Mathematical Risk Analysis: Dependence, Risk
  Bounds, Optimal Allocations and Portfolios}}, Springer, Berlin, 2013.

\bibitem[Wan16]{Wang}
Ruodu Wang, \emph{{Regulatory arbitrage of risk measures}}, Quantitative
  Finance \textbf{16} (2016), no.~3, 337--347.

\bibitem[Web18]{WEBER2018191}
Stefan Weber, \emph{{Solvency {II}, or how to sweep the downside risk under the
  carpet}}, Insurance: Mathematics and Economics \textbf{82} (2018), 191--200.

\bibitem[WW20]{https://doi.org/10.1111/mafi.12270}
Ruodu Wang and Yunran Wei, \emph{{Risk functionals with convex level sets}},
  Mathematical Finance \textbf{30} (2020), no.~4, 1337--1367.

\bibitem[ZH24]{ZOU20241}
Zhenfeng Zou and Taizhong Hu, \emph{Adjusted higher-order expected shortfall},
  Insurance: Mathematics and Economics \textbf{115} (2024), 1--12.

\bibitem[ZWXH23]{ZOU2023255}
Zhenfeng Zou, Qinyu Wu, Zichao Xia, and Taizhong Hu, \emph{{Adjusted Rényi
  entropic Value-at-Risk}}, European Journal of Operational Research
  \textbf{306} (2023), no.~1, 255--268.

\end{thebibliography}

\FloatBarrier

\appendix
\section{Further explanations to Table~\ref{table1}}\label{sec:explanationsTable}
Regarding the risk measures in Table~\ref{table1}, we give explanations and counterexamples for convexity, positive homogeneity, subadditivity, consistency with SSD and surplus invariance.

\textit{Convexity:} SCRM and CRM are not convex. To prove this, we create a counterexample. Let $\varepsilon=\frac{1}{100}$ and $c>0$. Then set $g(p)=(c+\varepsilon)\mathds{1}_{\left[\frac{2}{3},1\right]}(p).$ Due to the assumption of an atomless probability space, there exist i.i.d.~random variables $X,Y \in L^1$ with $\mathbb{P}(X=c)=\mathbb{P}(X=0)=\mathbb{P}(X=-\frac{1}{2}c)=\frac{1}{3}$, see~\cite[Proposition A.31]{follmer_stochastic_2016}. Then, it holds that $\mathbb{P}\left(\frac{1}{2}X+\frac{1}{2}Y=c\right)=\mathbb{P}\left(\frac{1}{2}X+\frac{1}{2}Y=0\right)=\mathbb{P}\left(\frac{1}{2}X+\frac{1}{2}Y=-\frac{1}{2}c\right)=\frac{1}{9}$ and $\mathbb{P}\left(\frac{1}{2}X+\frac{1}{2}Y=c\right)=\mathbb{P}\left(\frac{1}{2}X+\frac{1}{2}Y=\frac{1}{4}c\right)=\mathbb{P}\left(\frac{1}{2}X+\frac{1}{2}Y=-\frac{1}{4}c\right)=\frac{2}{9}$. We obtain that
	\begin{equation*} 
		\scrm\left(\frac{1}{2}X+\frac{1}{2}Y\right)\geq 
		\begin{cases} 
			\text{ES}_{\frac{4}{9}+\varepsilon}\left(\frac{1}{2}X+\frac{1}{2}Y\right)-g\left(\frac{4}{9}+\varepsilon\right) & \text{if}\ \frac{4}{9}+\varepsilon>r, \\
			\text{VaR}_{\frac{4}{9}+\varepsilon}(\frac{1}{2}X+\frac{1}{2}Y)-g\left(\frac{4}{9}+\varepsilon\right) & \text{if}\ \frac{4}{9}+\varepsilon\leq r. \\
		\end{cases} 
	\end{equation*}
    Then, $g\left(\frac{4}{9}+\varepsilon\right) = 0$ gives us $\scrm\left(\frac{1}{2}X+\frac{1}{2}Y\right)>0$. Hence, for $r>\frac{2}{3}$ we get
	$\scrm\left(\frac{1}{2}(X+Y)\right)>0 =\VaR_{\frac{2}{3}}(X) = \scrm(X)=\frac{1}{2}\scrm(X)+\frac{1}{2}\scrm(Y).$
	So, the SCRM is not convex in general.
 
	   For the counterexample of the CRM, the underlying levels $\mathcal{L} = \{p_k\}_{k\in\{0,1,\dots n,,n+1\}}$ of the CRM are chosen such that $p_{n-1}=\frac{4}{9}$ and $p_{n}=\frac{2}{3}$. Then, we get that
	$\crm\left(\frac{1}{2}X+\frac{1}{2}Y\right)\geq \text{RVaR}_{\frac{4}{9},\frac{2}{3}}\left(\frac{1}{2}X+\frac{1}{2}Y\right)-g\left(\frac{4}{9}\right)>0$
	and
	$\frac{1}{2}\crm(X)+\frac{1}{2}\crm(Y) = \crm(X)=\VaR_{\frac{2}{3}}(X) =0.$
	So, the CRM is not convex in general.
    In a similar manner, we can create a counterexample for the convexity of the AERM, using the concavity of the expectiles for a level below $\frac{1}{2}$.
	
	\textit{Positive homogeneity:} By Theorem \ref{theo} we know that SCRM and CRM are  positive homogeneous if and only if $g(p)\in(0,\infty)$ for at most one $p\in(0,1]$. For the AERM we construct a counterexample with a step function $g$. To do so, let $\varepsilon=\frac{1}{100}$ and set $g(p)= (2-\varepsilon)\mathds{1}_{\left(0,\frac{2}{3}\right]}(p)+\infty\cdot\mathds{1}_{\left(\frac{2}{3},1\right]}(p).$
	By the assumption of an atomless probability space, there exist a random variable $X\in L^1$ with $\mathbb{P}(X=4) = \mathbb{P}(X=-2)=\frac{1}{2}$. Then, Example~\ref{example} gives us that 
	$\aerm(X)=\max\left\{e_0(X),e_{\frac{2}{3}}(X)-g\left(\frac{2}{3}\right)\right\} = \max\{-2,2-(2-\varepsilon)\}=\varepsilon$
	and hence,
	$2\aerm(X) = 2 \varepsilon< 2+\varepsilon=4-(2-\varepsilon)=e_{\frac{2}{3}}(2 X)-g\left(\frac{2}{3}\right)=\aerm(2X).$
	So, the AERM is not positive homogeneous in general.
	
	\textit{Subadditivity:} Now, we show that the SCRM and the CRM are not subadditive. Let $c>0$, $\varepsilon=\frac{1}{100}$, $q \in(\varepsilon,1)$,  and set  
    $g(p) = \frac{3}{2}c\mathds{1}_{[q-\varepsilon,1]}(p).$	Further, let $X\in L^1$ be a random variable such that $\mathbb{P}(X=0)=q-\varepsilon$ and $\mathbb{P}(X=c)=1 -(q-\varepsilon)$. For the SCRM we obtain that  
	\begin{equation*} 
		\scrm(X+X)\geq 
		\begin{cases} 
			\text{ES}_{q}(2X)-g(q) & \text{if}\ q>r, \\
			\text{VaR}_{q}(2X)-g(q) & \text{if}\ q\leq r. \\
		\end{cases} 
	\end{equation*}
	Hence, by $c<g(q-\varepsilon)\leq g(q)<2c$ we conclude for $q\leq r$ that
	$\scrm(X+X)>0 = 2\cdot 0 =2\scrm(X).$
	Thus, the SCRM is not subadditive in general. If we now set $p_{n-1} = q-\varepsilon$, then with a similar $g$ and $X$ as before, we get that 
	$\crm(X+X)\geq\text{RVaR}_{q,p_{n-1}}(2X)-g(q)>0= 2\cdot 0 =2\crm(X).$
	So, also the CRM is not subadditive in general.
 
	The AERM is also not subadditive. Indeed, from the aforementioned argumentation for the positive homogeneity, we see that 
	$2\aerm(X) < \aerm(2X),$
	which contradicts the subadditivity.

	\textit{Consistency with SSD:} We know that the VaR and RVaR are not consistent with SSD. For the case of RVaR see~\cite[Remark 9]{RVaR2} or~\cite{RVaR3}.
    Since VaR, respectively RVaR, are special cases of SCRM, respectively CRM, we obtain that SCRM and CRM are not consistent with SSD in general. The AERM is consistent with SSD iff the target risk profile is zero for all level smaller than $\frac{1}{2}$, because it is  convex and law-invariant~\cite[Proposition $3.2$]{Mao}. 
	
	\textit{Surplus invariance:} SCRM and CRM are not surplus invariant. For brevity, we only look at the SCRM. Let $q,s\in\left(0,\frac{1}{2}\right)$ with $s<q$. Then, let $X\in L^1$ such that $\mathbb{P}(X=-c)=1-\mathbb{P}(X=c)=q$. Further, set $g(p) = \infty\cdot \mathds{1}_{\left(\frac{q+s}{2},1\right]}(p)$. Then, we obtain for $r<\frac{q+s}{2}$ that
    $\scrm(X)=\text{ES}_{\frac{q+s}{2}}(X)=c(1-q)-c \left(q-\frac{q+s}{2}\right)<c(1-q) = \scrm(\max\{X,0\}).$
 
	The AERM is also not surplus invariant, because the entire tail is used to compute the its value. For a counterexample, let   
    $g(p) = \infty\cdot\mathds{1}_{\left(\frac{1}{2},1\right]}(p)$ and $X\in L^1$ with  $\mathbb{P}(X=-2c)=\mathbb{P}(X=2c)=\frac{1}{2}$ for $c>0$. Then, Example~\ref{example} gives us that
	$\aerm(X)=e_{\frac{1}{2}}(X)=0<c=e_{\frac{1}{2}}(\max\{X,0 \})\leq\aerm(\max\{X,0\}).$
	This shows that the AERM is not surplus invariant in general.

\section{Figures}\label{sec:figures}
Here we present additional figures for Section~\ref{sub bench}. In doing so, Figure~\ref{Lowvolag} illustrates the benchmark profiles of the adjusted ES as well as the AERM for low and high volatility frames. Figure~\ref{Stockbenchmarkaerm1} shows the AERM for different benchmark profiles, analogously to Figure~\ref{10}.
	\begin{figure}[h]
		\begin{center}
			\includegraphics[scale=0.66]{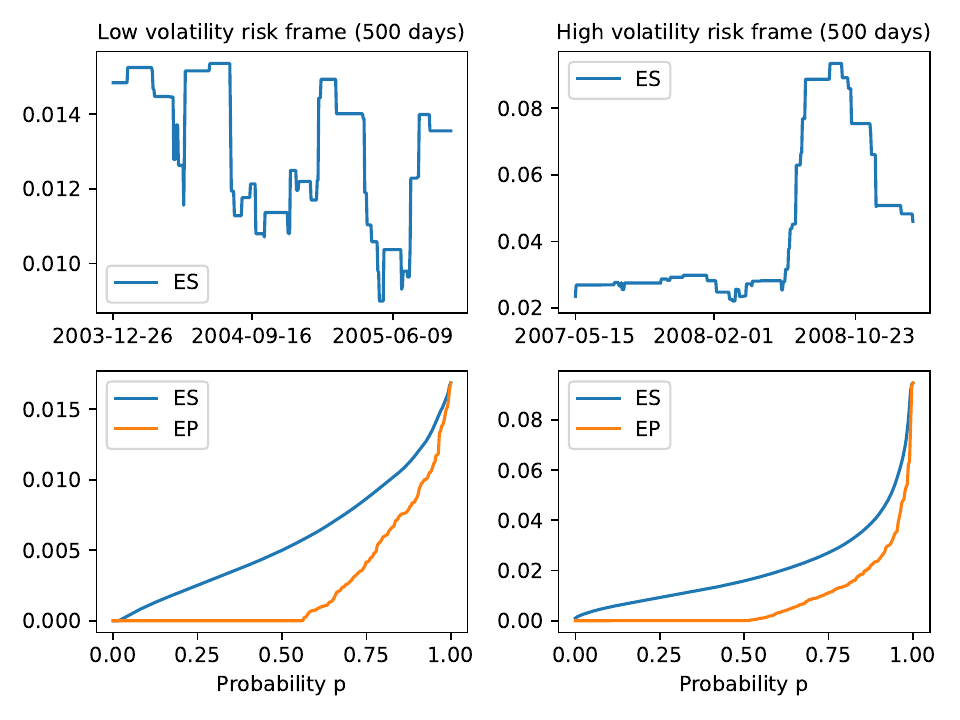}
		\end{center}
		\caption{\footnotesize{{\textit{Upper plots:} $95\%$ ES calculated for a 60-day rolling window for the S$\&$P 500 in a low or high volatility time frame. \textit{Lower plots:} Target risk profiles based on ES and expectiles based on the low or high volatility time frame (500 days).}}}
		\label{Lowvolag}
	\end{figure}

	\begin{figure}[p]
		\begin{center}
			\includegraphics[scale=0.75]{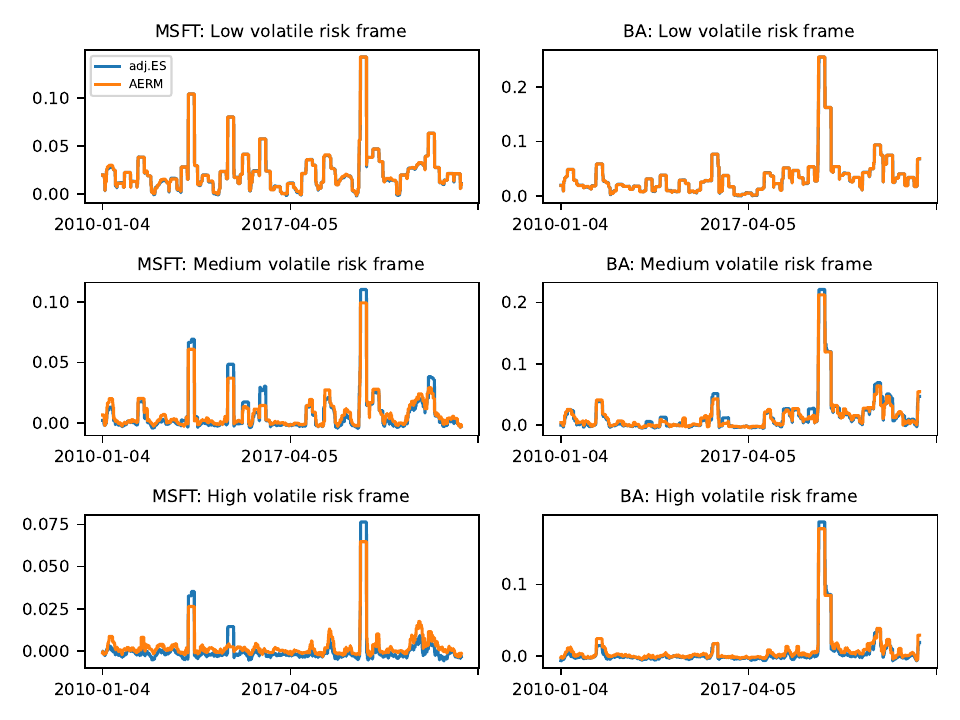}
		\end{center}
        \caption{\footnotesize{{AERM and adjusted ES (adj.~ES) plotted for  Microsoft (MSFT) and Boeing (BA) for three different benchmark profiles, which are based on the S$\&$P $500$.}}}
		\label{Stockbenchmarkaerm1}
	\end{figure}
\end{document}